\RequirePackage{amsmath}
\documentclass{jalc}

\usepackage{multicol}
\usepackage{amsfonts}				
\usepackage{amssymb}
\usepackage{graphicx}
\usepackage{epic}
\usepackage{eepic}
\usepackage{epsfig,float}
\usepackage{verbatim}
\usepackage{pdfsync}
\usepackage{gastex}
\usepackage{bbm}

\usepackage{xcolor}

\renewcommand{\le}{\leqslant}
\renewcommand{\ge}{\geqslant}

\newcommand{\ol}{\overline}
\newcommand{\eps}{\varepsilon}
\newcommand{\emp}{\emptyset}

\newcommand{\Sig}{\Sigma}

\newcommand{\noin}{\noindent}

\newcommand{\bi}{\begin{itemize}}
\newcommand{\ei}{\end{itemize}}
\newcommand{\be}{\begin{enumerate}}
\newcommand{\ee}{\end{enumerate}}
\newcommand{\bd}{\begin{description}}
\newcommand{\ed}{\end{description}}
\newcommand{\bq}{\begin{quote}}
\newcommand{\eq}{\end{quote}}

\newcommand{\cD}{{\mathcal D}}

\newcommand{\cN}{{\mathcal N}}

\newcommand{\cP}{{\mathcal P}}

\newcommand{\cU}{{\mathcal U}}

\newcommand{\lraL}{{\mathbin{\approx_L}}}

\begin{document}

\title{Unrestricted State Complexity of Binary Operations on  Regular and Ideal Languages}

\runningtitle{Unrestricted State Complexity}
\runningauthors{ \textsc{J.~Brzozowski}, \textsc{C.~Sinnamon} }
  
\author[WAT,NSERC]{Janusz~Brzozowski}
\address[WAT]{David R. Cheriton School of Computer Science, University of Waterloo, \\
Waterloo, ON, Canada N2L 3G1\\
\email[J.~Brzozowski]{brzozo@uwaterloo.ca}
\email[C.~Sinnamon]{sinncore@gmail.com}
}
\thanks[NSERC]{This work was supported by the Natural Sciences and Engineering Research Council of Canada 
grant No.~OGP0000871.}

\author[WAT,NSERC]{Corwin Sinnamon}

\maketitle

\begin{abstract}
We study the state complexity of binary operations on regular languages over different alphabets. It is  known that if $L'_m$ and $L_n$ are languages of state complexities  $m$ and $n$, respectively, and restricted to the same alphabet,   the state complexity of any binary boolean operation on $L'_m$ and $L_n$ is $mn$, and that of  product (concatenation) is $m 2^n - 2^{n-1}$. 
In contrast to this, we show that if $L'_m$ and $L_n$ are over different alphabets, 
the state complexity of union and symmetric difference is $(m+1)(n+1)$, 
that of difference is $mn+m$, 
that of intersection  is $mn$,
and that of product is  $m2^n+2^{n-1}$.
We also study unrestricted complexity of binary operations in the classes of regular right, left, and two-sided ideals, and derive tight upper bounds. 
The bounds for product of the unrestricted cases (with the bounds for the restricted cases in parentheses) are as follows: 
right ideals $m+2^{n-2}+2^{n-1}$ ($m+2^{n-2}$); 
left ideals $mn+m+n$ ($m+n-1$);
two-sided ideals $m+2n$ ($m+n-1$).
The state complexities of boolean operations on all three types of ideals are the same as those of arbitrary regular languages, whereas that is not the case if the alphabets of the arguments are the same.
Finally, we update the known results about most complex regular, right-ideal, left-ideal, and two-sided-ideal languages to include the unrestricted cases.
\medskip

\noin
\end{abstract}

\section{Motivation}

Formal definitions are postponed until Section~\ref{sec:basics}. 
\smallskip

The first comprehensive  paper on state complexity was published in 1970 by A.~N.~Maslov~\cite{Mas70}, but this work was unknown in the West for many years. Maslov wrote:  
 \begin{quote} {\it An important measure of the complexity of [sets of words representable in finite automata] is the number of states in the minimal representing automaton.
... if $T(A) \cup T(B)$ are representable in automata $A$ and $B$ with $m$ and $n$ states respectively ..., then:
	\be 
	\item
	$T(A) \cup T(B)$ is representable in an automaton with $m\cdot n$ states;
	\item
	$T(A).T(B)$ is representable in an automaton with $(m-1)\cdot2^n + 2^{n-1}$ states.
	\ee}
\end{quote}

In this formulation these statements are false: we will show that union may require $(m+1)(n+1)$ states and product (concatenation), $m2^n+2^{n-1}$ states. However, Maslov must have had in mind languages over the same alphabet, in which case the statements are correct.

The second comprehensive paper on state complexity was published by S.~Yu, Q.~Zhuang and K.~Salomaa~\cite{YZS94} in 1994. Here the authors wrote:
\begin{quote}{\it
\be
\item
... for any pair of complete $m$-state DFA $A$ and $n$-state DFA $B$ defined on the same alphabet $\Sigma$, there exists a DFA with at most $m2^n-2^{n-1}$ states which accepts $L(A)L(B)$.
\item
... $m\cdot n$ states are ... sufficient ... for a DFA to accept the intersection (union) of an $m$-state DFA language and an $n$-state DFA language.
\ee}
\end{quote}

 The first statement includes the same-alphabet restriction, but the second omits it (presumably it is implied by the context). 
Here DFA stands for \emph{deterministic finite automaton}, and \emph{complete} means that there is a transition from every state under every input letter.

After these two papers appeared many authors studied the state complexity of various operations in various classes of regular languages, always using witnesses restricted to the same alphabet.
However, we point out that the  same-alphabet restriction is unnecessary: 
there is no reason why we should not  compute  the union or product of two languages over different alphabets.
In fact, the software package \emph{Grail},  for instance,  ({\small \tt http://www.csit.upei.ca/theory/}) allows the user to calculate the result of these operations. 

As an example, let us consider the union of 
languages $L'_2=\{a,b\}^*b$ and $L_2=\{a,c\}^*c$  accepted by the minimal complete two-state automata $\cD'_2$ and $\cD_2$ of Figure~\ref{fig:example}, where an incoming arrow denotes the initial state and a double circle represents a final state.

\begin{figure}[ht]
\unitlength 9pt
\begin{center}\begin{picture}(30,4.3)(0,6.5)
\gasset{Nh=2,2,Nw=2.2,Nmr=1.25,ELdist=0.4,loopdiam=1.5}
{\small
\node(0')(1,7){$0'$}\imark(0')
\node(1')(8,7){$1'$}\rmark(1')
\node(0)(22,7){0}\imark(0)
\node(1)(29,7){1}\rmark(1)
\drawloop(0'){$a$}
\drawloop(1'){$b$}
\drawedge[curvedepth= .8,ELdist=.4](0',1'){$b$}
\drawedge[curvedepth= .8,ELdist=.4](1',0'){$a$}
\drawloop(0){$a$}
\drawloop(1){$c$}
\drawedge[curvedepth= .8,ELdist=.4](0,1){$c$}
\drawedge[curvedepth= .8,ELdist=.4](1,0){$a$}
}
\end{picture}\end{center}
\caption{Two minimal complete DFAs $\cD'_2$ and $\cD_2$.}
\label{fig:example}
\end{figure}
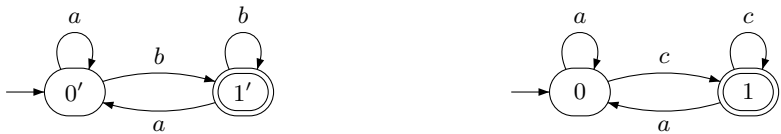
The union of $L'_2$ and $L_2$ is a language over three letters. To find the DFA for $L'_2 \cup L_2$, we view $\cD'_2$ and $\cD_2$ as incomplete DFAs, the first missing all transitions under $c$, and the second, under $b$. 
After adding the missing transitions we obtain DFAs $\cD'_3$ and $\cD_3$ shown in Figure~\ref{fig:complete}. Now we can proceed as is usually done in the same-alphabet approach, and  use the direct product of 
$\cD'_3$ and $\cD_3$ to find $L_2'\cup L_2$. Here it turns out that six states are necessary to represent $L'_2\cup L_2$, but the state complexity of  union is actually 
$(m+1)(n+1)$. 

\begin{figure}[ht]
\unitlength 9pt
\begin{center}\begin{picture}(30,9)(0,1.6)
\gasset{Nh=2.2,Nw=2,2,Nmr=1.25,ELdist=0.4,loopdiam=1.5}
{\small
\node(0')(1,7){$0'$}\imark(0')
\node(1')(8,7){$1'$}\rmark(1')
\node(2')(4.5,3){$2'$}
\node(0)(22,7){0}\imark(0)
\node(1)(29,7){1}\rmark(1)
\node(2)(25.5,3){2}
\drawloop(0'){$a$}
\drawloop(1'){$b$}
\drawedge[curvedepth= .8,ELdist=.4](0',1'){$b$}
\drawedge[curvedepth= .8,ELdist=.4](1',0'){$a$}
\drawloop[loopangle=270,ELpos=25](2'){$a,b,c$}
\drawloop(0){$a$}
\drawloop(1){$c$}
\drawedge[ELdist=-1.1](0',2'){$c$}
\drawedge[ELdist=.3](1',2'){$c$}
\drawedge[curvedepth= .8,ELdist=.4](0,1){$c$}
\drawedge[curvedepth= .8,ELdist=.4](1,0){$a$}
\drawedge[ELdist=-1.1](0,2){$b$}
\drawedge[ELdist=.3](1,2){$b$}
\drawloop[loopangle=270,ELpos=25](2){$a,b,c$}
}
\end{picture}\end{center}
\caption{DFAs $\cD'_3$ and $\cD_3$ over three letters.}
\label{fig:complete}
\end{figure}
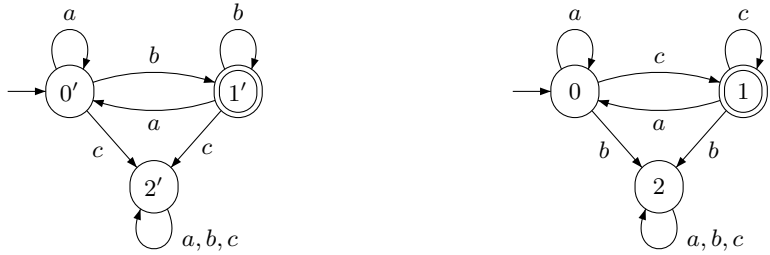

In general, when calculating the result of a binary operation on regular languages with different alphabets, we deal with special incomplete DFAs that are only missing some letters and all the transitions caused by these letters. The complexity of incomplete DFAs has been studied previously by Gao, K. Salomaa, and Yu~\cite{GSY11} and by Maia, Moreira and Reis~\cite{MMR15}. However, the objects studied there are \emph{arbitrary} incomplete DFAs, whereas we are interested only in \emph{complete DFAs with some missing letters}. 
Secondly, we study \emph{state} complexity, whereas the above-mentioned papers deal mainly with \emph{transition} complexity.
Nevertheless, there is some overlap. It was shown in~\cite[Corollary 3.2]{GSY11} that the incomplete state complexity of union is less than or equal to $mn+m+n$, and that this bound is tight in some special cases. In~\cite[Theorem 2]{MMR15}, witnesses that work in all cases were found. These complexities correspond to our result for union in Theorem~\ref{thm:boolean}.
Also in~\cite[Theorem 5]{MMR15}, the incomplete state complexity of product is shown to be $m2^n+2^{n-1}-1$, and this corresponds to our result for product in Theorem~\ref{thm:product}.

In this paper we remove the restriction of equal alphabets of the two operands. 
We prove that
the complexity of union and symmetric difference is $(m+1)(n+1)$, 
that of difference is $mn+m$, and that of intersection is $mn$,                                                                                                     
and that of the product is  $m2^n+2^{n-1}$, if each language's own alphabet is used. 
We exhibit a new most complex regular language that meets the complexity bounds for  restricted and unrestricted boolean operations, restricted and unrestricted products, star, and reversal, has a maximal syntactic semigroup and most complex atoms.  All the witnesses used here are derived from that one most complex language.

A much shorter version of this paper appeared in~\cite{Brz16}. 
That paper dealt only with  unrestricted product and binary boolean operations on regular languages.
Here we include a shorter proof of the theorem about unrestricted product of regular languages, and  establish the unrestricted complexities of product and binary boolean operations on right, left  and two-sided ideals.

\section{Terminology and Notation}
\label{sec:basics}

If $\Sig$ is a finite alphabet and $L\subseteq \Sig^*$, the \emph{alphabet of $L$} is the set
$\Sig_L=\{a \in \Sig\mid uav\in L \text{ for some } u,v\in\Sig^*\}$.
%
A basic complexity measure of $L$ with alphabet $\Sig_L$ is the number $n$ of distinct (left) quotients of $L$ by words in $\Sig_L^*$, where a \emph{(left) quotient} of $L$ by a word $w\in\Sig_L^*$ is $w^{-1}L=\{x\mid wx\in L\}$. 
The number of quotients of $L$ is its \emph{quotient complexity}~\cite{Brz10a}, $\kappa(L)$.

Unless otherwise specified, for a regular language $L$ with alphabet $\Sigma_L$ we define the complement of $L$ by $\ol{L} = \Sigma_L^* \setminus L$.
With this definition it is not always true that $\kappa(L)=\kappa(\ol{L})$ as $L$ and $\ol{L}$ may have different alphabets.
For example, if $L = \{a,b\}^* \setminus a^*$, then $\kappa(L) = 2$ while $\kappa(\ol{L}) = \kappa(a^*) = 1$, for the alphabet of $\ol{L}$ is $\{a\}$ instead of $\{a,b\}$.
There is only one way for this to occur:
In order for their alphabets to be different, there must be a letter in the alphabet of $L$ such that every word containing the letter is in the language, so that the letter is not present in $\ol{L}$.
Hence we have $\kappa(\ol{L}) \in \{ \kappa(L), \kappa(L)-1\}$, and it is usually easy to determine the complexity of $\ol{L}$ when presented with a specific language $L$.

Let $L_n$ be regular language with quotient complexity $n$, let $\circ$ be a unary operation on languages, and let $L_{ n }^\circ$ be the result of the operation.
The quotient complexity of the operation $\circ$ is the maximal value of $\kappa(L_{ n}^\circ)$ as a function of $n$, as $L_n$ ranges over all regular languages with quotient complexity $n$. 

Let $L'_m$ and $L_n$  be regular languages of  quotient complexities $m$ and $n$ that have alphabets $\Sig'$ and $\Sig$, respectively, let $\circ$ be a binary operation on languages, and let $L'_m \circ L_n$  be the result of the operation.
The \emph{quotient complexity of  $\circ$} 
is  the maximal value of $\kappa(L'_m \circ L_n)$ as a function of $m$ and $n$,
as $L'_m$ and $L_n$ range over all regular languages of quotient complexities $m$ and $n$, respectively.


%

A \emph{deterministic finite automaton (DFA)} is a quintuple
$\cD=(Q, \Sigma, \delta, q_0,F)$, where
$Q$ is a finite non-empty set of \emph{states},
$\Sig$ is a finite non-empty \emph{alphabet},
$\delta\colon Q\times \Sig\to Q$ is the \emph{transition function},
$q_0\in Q$ is the \emph{initial} state, and
$F\subseteq Q$ is the set of \emph{final} states.
We extend $\delta$ to a function $\delta\colon Q\times \Sig^*\to Q$ as usual.
A~DFA $\cD$ \emph{accepts} a word $w \in \Sigma^*$ if ${\delta}(q_0,w)\in F$. The language accepted by $\cD$ is denoted by $L(\cD)$. If $q$ is a state of $\cD$, then the language $L^q$ of $q$ is the language accepted by the DFA $(Q,\Sigma,\delta,q,F)$. 
A state is \emph{empty} (or \emph{dead} or a \emph{sink state}) if its language is empty. Two states $p$ and $q$ of $\cD$ are \emph{equivalent} if $L^p = L^q$. 
A state $q$ is \emph{reachable} if there exists $w\in\Sig^*$ such that $\delta(q_0,w)=q$.
A DFA  is \emph{minimal} if all of its states are reachable and no two states are equivalent.
Usually DFAs are used to establish upper bounds on the complexity of operations, and also as witnesses that meet these bounds.

The \emph{state complexity}~\cite{YZS94}  of a regular language $L$ is the number of states in a complete minimal DFA with alphabet $\Sig_L$ which recognizes the language.
This  concept is equivalent to quotient complexity of $L$.
For example, the state complexity of the language $a^*$ is one. 
There is a two-state minimal DFA with alphabet $\{a,b\}$ accepting $a^*$, but its alphabet is not $\Sig_L$.

Since we do not use any other measures of complexity in this paper (with the exception of one mention of time and space complexity in this paragraph), we refer to quotient/state complexity simply as \emph{complexity}.
The quotient/state complexity of an operation gives a worst-case lower bound on the time and space complexities of the operation. For this reason it has been studied extensively; see~\cite{Brz10a,Brz13,Yu01,YZS94} for additional references.

If $\delta(q,a)=p$ for a state  $q\in Q$ and a letter $a\in \Sig$, we say there is a \emph{transition} under $a$ from $q$ to $p$ in $\cD$.
The DFAs defined above are \emph{complete} in the sense that there is \emph{exactly one} transition for each state  $q\in Q$ and each letter $a\in \Sig$. If there is \emph{at most one transition}  for each   $q\in Q$ and $a\in \Sig$, the automaton is an \emph{incomplete} DFA.

A \emph{nondeterministic finite automaton (NFA)} is a 5-tuple
$\cD=(Q, \Sigma, \delta, I,F)$, where
$Q$,
$\Sig$ and $F$ are defined as in a DFA, 
$\delta\colon Q\times \Sig\to 2^Q$ is the \emph{transition function}, and
$I\subseteq Q$ is the \emph{set of initial states}. 
An \emph{$\eps$-NFA} is an NFA in which transitions under the empty word $\eps$ are also permitted.

To simplify the notation, without loss of generality we use $Q_n=\{0,\dots,n-1\}$ as our basic set of $n$ elements.
A \emph{transformation} of $Q_n$ is a mapping $t\colon Q_n\to Q_n$.
The \emph{image} of $q\in Q_n$ under $t$ is denoted by $qt$.
For $k\ge 2$, a transformation (permutation) $t$ of a set $P=\{q_0,q_1,\ldots,q_{k-1}\} \subseteq Q_n$ is a \emph{$k$-cycle}
if $q_0t=q_1, q_1t=q_2,\ldots,q_{k-2}t=q_{k-1},q_{k-1}t=q_0$.
This $k$-cycle is denoted by $(q_0,q_1,\ldots,q_{k-1})$, and acts as the identity on the states in $Q_n\setminus P$.
A~2-cycle $(q_0,q_1)$ is called a \emph{transposition}.
 A transformation that changes only one state $p$ to a state $q\neq p$ and acts as the identity for the other states is denoted by $(p\to q)$.
 The identity transformation is denoted by $\mathbbm 1$.
If $s,t$ are transformations of $Q_n$, their composition  when applied to $q \in Q_n$ is defined by
$(qs)t$.
The set of all $n^n$ transformations of $Q_n$ is a monoid under composition. 

We use $Q_n$ as the set of states of every DFA with $n$ states, and 0 as the initial state.
In any DFA $\cD_n=(Q_n, \Sigma, \delta, 0,F)$  each $a\in \Sig$ induces a transformation $\delta_a$ of  $Q_n$ defined by $q\delta_a=\delta(q,a)$; we denote this by $a\colon \delta_a$. 
For example, when defining the transition function of a DFA, we write $a\colon (0,1)$ to mean that
$\delta(q,a)=q(0,1)$, where the transformation $(0,1)$ acts on state $q$ as follows: if $q$ is 0  it maps it to 1, if $q$ is 1 it maps it to 0,
and it acts as the identity on the remaining states.

By a slight abuse of notation we use the letter $a$ to denote the transformation it induces; thus we write $qa$ instead of $q\delta_a$.
We extend the notation to sets of states: if $P\subseteq Q_n$, then $Pa=\{pa\mid p\in P\}$.
We also find it convenient to write $P\stackrel{a}{\longrightarrow} Pa$ to indicate that the image of $P$ under $a$ is $Pa$.

We extend these notions to arbitrary words. For each word $w \in \Sig^*$, the transition function induces a transformation $\delta_w$ of $Q_n$ by  $w$: for all $q \in Q_n$, 
$q\delta_w = \delta(q, w).$ 
The set $T_{\cD_{n}}$ of all such transformations by non-empty words is the \emph{transition semigroup} of $\cD_{n}$ under composition~\cite{Pin97}. 

The \emph{Myhill congruence} $\lraL_n$~\cite{Myh57} (also known as the \emph{syntactic congruence})  of a language $L_n \subseteq \Sig^*$ is defined on $\Sig^+$ as follows:
For $x, y \in \Sig^+,  x \, \lraL_n \, y $  if and only if  $wxz\in L_n  \Leftrightarrow wyz\in L_n$ for all  $w,z \in\Sig^*.
$
The quotient set $\Sig^+/ \lraL_n$ of equivalence classes of  $\lraL_n$ is a semigroup, the \emph{syntactic semigroup} $T_{L_n}$ of $L_n$.

If  $\cD_n$ is a minimal DFA of $L_n$, then $T_{\cD_n}$ is isomorphic to the syntactic semigroup $T_{L_n}$ of $L_n$~\cite{Pin97}, and we represent elements of $T_{L_n}$ by transformations in~$T_{\cD_n}$. 
The size of this semigroup has been used as a measure of complexity~\cite{Brz13,BrYe11,HoKo04,KLS05}.

The \emph{atom congruence} is a left congruence defined as follows: two words $x$ and $y$ are equivalent if 
 $ux\in L$ if and only if  $uy\in L$ for all $u\in \Sig^*$. 
 Thus $x$ and $y$ are equivalent if
 $x\in u^{-1}L$ if and only if $y\in u^{-1}L$.
 An equivalence class of this relation is called an \emph{atom} of $L$~\cite{BrTa14,Iva16}. 
It follows that an atom is a non-empty intersection of complemented and uncomplemented quotients of $L$. 
The number of atoms and their quotient complexities are possible measures of complexity of regular languages~\cite{Brz13}.
For more information about atoms and their complexity, see~\cite{BrTa13,BrTa14,Iva16}.

A sequence $(L_n, n\ge k)=(L_k,L_{k+1},\dots)$, of regular languages is called a \emph{stream}; here $k$ is usually some small integer, and the languages in the stream usually have the same form and differ only in the parameter $n$. For example, $(\{a,b\}^*a^n\{a,b\}^* \mid n\ge 2)$ is a stream.
To find the complexity of a binary operation $\circ$ we need to find an upper bound on this complexity and two streams 
 $(L'_m, m \ge h)$ and $(L_n, n\ge k)$ of languages meeting this bound.
 In general, the two streams are different, but there are many examples where $L'_n$ ``differs only slightly'' from $L_n$; such a language $L'_n$ is called a \emph{dialect}~\cite{Brz13} of $L_n$, and is defined below. 

Let $\Sig=\{a_1,\dots,a_k\}$ be an alphabet; we assume that its elements are ordered as shown.
Let $\pi$ be a \emph{partial permutation} of $\Sig$, that is, a partial function $\pi \colon \Sig \rightarrow \Gamma$ where $\Gamma \subseteq \Sig$, for which there exists $\Delta \subseteq \Sig$ such that $\pi$ is bijective when restricted to $\Delta$ and  undefined on $\Sig \setminus \Delta$. 
We denote undefined values of $\pi$ by  ``$-$'', that is, we write $\pi(a)=-$, if $\pi$ is undefined at $a$.

If $L\subseteq \Sig^*$, we denote it by $L(a_1,\dots,a_k)$ to stress its dependence on $\Sig$.
If $\pi$ is a partial permutation, let $s_\pi(L(a_1,\dots,a_k))$ be the language obtained from $L(a_1,\dots,a_k)$ by the substitution  $s_\pi$ defined as follows: 
 for $a\in \Sig$, 
$a \mapsto \{\pi(a)\}$ if $\pi(a)$ is defined, and $a \mapsto \emp$ otherwise.
The \emph{permutational dialect}, or simply \emph{dialect},  of $L(a_1,\dots,a_k)$ defined by $\pi$ is the  language 
$L(\pi(a_1),\dotsc,\pi(a_k))= s_\pi(L(a_1,\dots,a_k))$.

Similarly,
let $\cD = (Q_n,\Sig,\delta,0,F)$ be a DFA; we denote it by
$\cD(a_1,\dots,a_k)$ to stress its dependence on $\Sig$.
If $\pi$ is a partial permutation, then the \emph{permutational dialect}, or simply \emph{dialect}, 
$\cD(\pi(a_1),\dotsc,\pi(a_k))$ of
$\cD(a_1,\dots,a_k)$ is obtained by changing the alphabet of $\cD$ from $\Sig$ to $\pi(\Sig)$, and modifying $\delta$ so that in the modified DFA 
$\pi(a_i)$ induces the transformation induced by $a_i$  in the original DFA.
One verifies that if the language $L(a_1,\dots,a_k)$ is accepted by DFA $\cD(a_1,\dots,a_k)$, then
$L(\pi(a_1),\dotsc,\pi(a_k))$ is accepted by $\cD(\pi(a_1),\dotsc,\pi(a_k))$.

If the letters for which $\pi$ is undefined are at the end of the alphabet $\Sig$, then they are omitted. For example,
if $\Sig=\{a,b,c,d\}$ and $\pi(a)=b$, $\pi(b)=a$, and $\pi(c)=\pi(d)=-$, then we write $L_n(b,a)$ for $L_n(b,a,-,-)$, etc.

A \emph{most complex stream } of regular language is one that, together with some dialect streams, meets 
 the complexity bounds for all boolean operations, product, star, and reversal, and has the largest syntactic semigroup and most complex atoms.
In looking for a most complex stream we try to use the smallest possible alphabet sufficient to meet all the bounds.
Most complex streams are useful in systems dealing with regular languages and finite automata. One would like to know the maximal sizes of automata that can be handled by the system. In view of the existence of most complex streams, one stream can be used to test all the operations.
     \medskip

\section{Regular Languages}

The DFA of Definition~\ref{def:regular} will be used for both product and boolean operations on regular languages; this DFA is the 4-input DFA called $\cU_n(a,b,c,d)$ in~\cite{Brz13},
where it was shown that  $\cU_n(a,b,c)$ is a ``universal witness'', that is, 
$(\cU_n(a,b,c)\mid n\ge 3)$ is a most complex regular stream  for all common restricted operations.
We now prove that $\cU_n(a,b,c,d)$ (renamed $\cD_n(a,b,c,d)$ below), together with some of its permutational dialects, is most complex  for both restricted and unrestricted operations.

\begin{definition}
\label{def:regular}
For $n\ge 3$, let $\cD_n=\cD_n(a,b,c,d)=(Q_n,\Sig,\delta_n, 0, \{n-1\})$, where 
$\Sig=\{a,b,c,d\}$, 
and $\delta_n$ is defined by the transformations
$a\colon (0,\dots,n-1)$,
$b\colon(0,1)$,
$c\colon(n-1 \rightarrow 0)$, and
$d\colon {\mathbbm 1}$.
Let $L_n=L_n(a,b,c,d)$ be the language accepted by~$\cD_n$.
The structure of $\cD_n(a,b,c,d)$ is shown in Figure~\ref{fig:RegWit}. 
\end{definition}

\begin{figure}[ht]
\unitlength 8.5pt
\begin{center}\begin{picture}(37,10)(0,2)
\gasset{Nh=1.9,Nw=3.6,Nmr=1.25,ELdist=0.4,loopdiam=1.5}
	{\small
\node(0)(1,7){0}\imark(0)
\node(1)(8,7){1}
\node(2)(15,7){2}
\node[Nframe=n](3dots)(22,7){$\dots$}
\node(n-2)(29,7){$n-2$}
\node(n-1)(36,7){$n-1$}\rmark(n-1)
\drawloop(0){$c,d$}
\drawedge[curvedepth= .8,ELdist=.1](0,1){$a,b$}
\drawedge[curvedepth= .8,ELdist=-1.2](1,0){$b$}
\drawedge(1,2){$a$}
\drawloop(2){$b,c,d$}
\drawedge(2,3dots){$a$}
\drawedge(3dots,n-2){$a$}
\drawloop(n-2){$b,c,d$}
\drawedge(n-2,n-1){$a$}
\drawedge[curvedepth= 4.0,ELdist=-1.0](n-1,0){$a,c$}
\drawloop(n-1){$b,d$}
\drawloop(1){$c,d$}
}
\end{picture}\end{center}
\caption{ DFA  of  Definition~\ref{def:regular}.}
\label{fig:RegWit}
\end{figure}
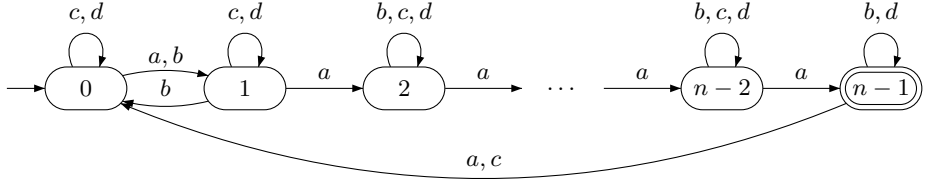

\subsection{Product of Regular Languages}

\begin{theorem}[(Product of Regular Languages)]
\label{thm:product}
For $m,n \ge 3$, let $L'_m$ (respectively, $L_n$) be a regular language with $m$ (respectively, $n$) quotients over an alphabet $\Sig'$, (respectively, $\Sig$). 
Then $\kappa(L'_m L_n) \le m2^n+2^{n-1}$, and this bound is met by $L'_m(a,b,-,c)$ and $L_n(b,a,-,d)$ of Definition~\ref{def:regular}.
\end{theorem}
\begin{proof}
First we derive the upper bound.
Let $\cD'_m=( Q'_m, \Sig', \delta', 0',F')$  and $\cD_n=(Q_n,\Sig,\delta,0,F) $ be minimal DFAs of arbitrary regular languages $L'_m$  and $L_n$, respectively.
We use the normal construction of an $\eps$-NFA $\cN$ to recognize $L'_mL_n$, by introducing an $\eps$-transition from each final state of $\cD'_m$ to the initial state of $\cD_n$, and changing all final states of $\cD'_m$ to non-final. This is illustrated in Figure~\ref{fig:product}, where $(m-1)'$ is the only final state of $\cD'_m$.
We then determinize $\cN$ using the subset construction to get the DFA $\cD$ for $L'_m L_n$.

Suppose $\cD'_m$ has $k$ final states, where $1 \le k \le m-1$.
We will show that $\cD$ can have only the following types of states: (a)  
at most $(m-k)2^n$ states $\{p'\} \cup S$, where $p'\in Q'_m\setminus F'$, and $S\subseteq Q_n$, 
(b) at most $k2^{n-1}$ states  
$\{ p' , 0\} \cup S$, where $p'\in F'$ and $S\subseteq Q_n\setminus \{0\}$,
and (c) at most $2^n$ states $S\subseteq Q_n$.
Because $\cD'_m$ is deterministic, there can be at most one state $p'$ of $Q'_m$ in any reachable subset. 
If $p' \notin F'$, it may be possible to reach any subset of states of $Q_n$ along with $p'$, and this accounts for (a).
If $p' \in F'$, then the set must contain $0$ and possibly any subset of $Q_n\setminus \{0\}$, giving (b).
It may also be possible to have any subset $S$ of $Q_n$ by applying an input that is not in $\Sig'$ to $\{0'\} \cup S$ to get $S$, and so we have~(c).
Altogether, there are at most $(m-k)2^n +k2^{n-1}+2^n = (2m-k)2^{n-1} + 2^n$ reachable subsets. This expression reaches its maximum when $k=1$, and so we have at most $m2^n+2^{n-1}$ states in $\cD$.

\begin{figure}[ht]
\unitlength 7.5pt
\begin{center}\begin{picture}(37,18)(-4,2)
\gasset{Nh=2.2,Nw=5,Nmr=1.25,ELdist=0.4,loopdiam=1.5}
	{\small
\node(0')(-4,14){$0'$}\imark(0')
\node(1')(3,14){$1'$}
\node(2')(10,14){$2'$}
\node[Nframe=n](3dots')(17,14){$\dots$}
\node(m-1')(24,14){$(m-1)'$}
	
\drawedge[curvedepth= 1.4,ELdist=-1.3](0',1'){$a,b$}
\drawedge[curvedepth= 1,ELdist=.3](1',0'){$b$}
\drawedge(1',2'){$a$}
\drawedge(2',3dots'){$a$}
\drawedge(3dots',m-1'){$a$}
\drawedge[curvedepth= -5.2,ELdist=-1](m-1',0'){$a$}
\drawloop(0'){$c$}
\drawloop(1'){$c$}
\drawloop(2'){$b,c$}
\drawloop(m-1'){$b,c$}
\gasset{Nh=2.2,Nw=5,Nmr=1.25,ELdist=0.4,loopdiam=1.5}
\node(0)(7,7){0}\imark(0)
\node(1)(14,7){1}
\node(2)(21,7){2}
\node[Nframe=n](3dots)(28,7){$\dots$}
\node(n-1)(35,7){$n-1$}\rmark(n-1)
	
\drawloop(0){$d$}
\drawloop(1){$d$}
\drawloop(2){$a,d$}
\drawloop(n-1){$a,d$}
\drawedge[curvedepth= 1.2,ELdist=-1.3](0,1){$a,b$}
\drawedge[curvedepth= .8,ELdist=.25](1,0){$a$}
\drawedge(1,2){$b$}
\drawedge(2,3dots){$b$}
\drawedge(3dots,n-1){$b$}
\drawedge[curvedepth= 4.4,ELdist=-1.2](n-1,0){$b$}
\drawedge[curvedepth= -1.5,ELdist=-1](m-1',0){$\eps$}
	}
\end{picture}\end{center}
\caption{An NFA  for the product of  $L'_m(a,b,-,c)$ and $L_n(b,a,-,d)$. }
\label{fig:product}
\end{figure}
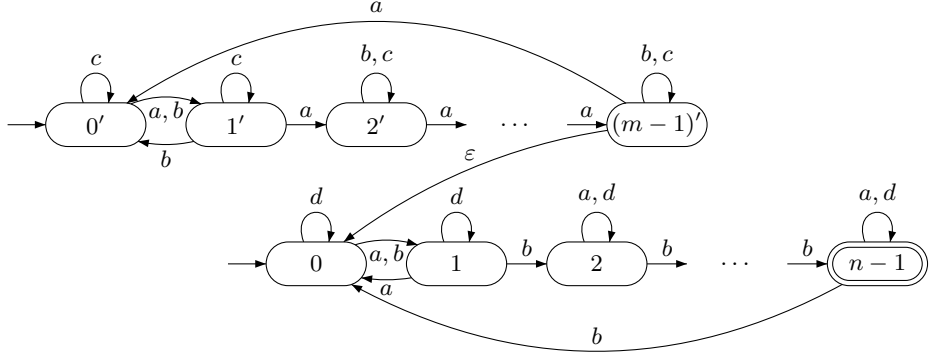

We prove that the bound is met by the  witnesses of Figure~\ref{fig:product}. We use the following result to show that all the states in the subset construction are reachable.

Suppose $\cD'_m=( Q'_m, \Sig, \delta', 0', \{f'\} )$ is a minimal DFA of $L'_m$, $f'\neq 0'$, and 
$\cD_n=(Q_n,\Sig,\delta,0,F) $ is a minimal DFA of  $L_n$.
Moreover, assume that 
the transition semigroups of $\cD'_m$ and $\cD_n$ are groups.

\begin{lemma}[(Sylvie Davies, personal communication)]
If all the sets of the form
$ \{p'\},\; p'\in Q_m' \setminus \{f'\},  \text{ and  } \{0',q\}, \; q \in Q_n $
are reachable, then so are all sets of the form 
$$ \{p'\} \cup S,\; p'\in Q_m' \setminus \{f'\}, 
\; S\subseteq Q_n \text{ and  } \{f',0\} \cup S, \; S \subseteq Q_n\setminus \{0\}. \quad \quad (*) $$
\end{lemma} 

We now prove that the conditions of the lemma apply to our case.
The initial state in the subset automaton is $\{0'\}$, state $\{p'\}$ is reached by $a^p$ if $p<m-1$, and $\{(m-1)', 0\}$ is reached by $a^{n-1}$.
Also, $\{0',1\}$ is reached by $a^m$.

\bi
\item
If $n$ is odd, from $\{0',1\}$ we reach $\{0',q\}$, $q\in Q_n$, by words in $(bb)^*$.
\item
If $n$ is even, from $\{0',1\}$ we reach $\{0',q\}$ with $q$ odd by words in $(bb)^*$.
\item
From $\{0',n-1\}$ we reach $\{0',0\}$ by $ab$. 
\item
From $\{0',0\}$ we reach  $\{0',q\}$ with $q$ even by $(bb)^*$.
\ei
Since  $ \{p'\},\; p'\in Q_m' \setminus \{f'\}$ and $\{0',q\}, \; q\in Q_n$ are reachable, so are all the sets of form $(*)$ by the Lemma.
\smallskip

For distinguishability, note that only state $q$ accepts $w_q = b^{n-1-q}$ in $\cD_n$. Hence, if two states of the product have different  sets $S$ and $S'$ and $q\in S\oplus S'$, then they can be distinguished by $w_q$. State $\{p'\} \cup S$ is distinguished from $S$ by $ca^{m-1-p}b^{n-1}$. If $p <q$, states $\{p'\} \cup S$ and $\{q'\} \cup S$ are distinguished as follows:Use $ca^{m-1-q}$ to reach $\{(p+m-1-q)'\}$ from $p'$ and $\{(m-1)', 0\}$ from $q'$. The reached states are distinguishable since they differ in their subsets of $Q_n$. 
\end{proof}

\subsection{Boolean Operations on Regular Languages}

Suppose $A, B \subseteq U$, where $U$ is some universal set.
A binary operation $\circ \colon \mathcal{P}(U) \times \mathcal{P}(U) \to \mathcal{P}(U)$ is \emph{boolean} if, for any $x \in U$, whether $x$ is included in $A \circ B$ depends only on the membership of $x$ in $A$ and $B$.
Thus there are sixteen binary boolean operations, corresponding to the number of ways of filling out the truth table below.
\begin{center}
\begin{tabular}{c|c|c}
$x \in A$ & $x \in B$ & $x \in A \circ B$\\
\hline
T & T & -\\
T & F & -\\
F & T & -\\
F & F & -
\end{tabular}
\end{center}
A boolean operation is \emph{proper} if it is not constant and does not depend on only one variable.
There are ten proper boolean operations, given below.
\begin{align*}
L'_m &\cup L_n                  &      \ol{L'_m} &\cap \ol{L_n} & &\\
\ol{L'_m} &\cup L_n          &       L'_m \cap \ol{L_n} &= L'_m \setminus L_n & L'_m&\oplus L_n\\
L'_m &\cup \ol{L_n}          &      \ol{L'_m} \cap L_n &= L_n \setminus L'_m  & L'_m &\oplus \ol{L_n}\\
\ol{L'_m} &\cup \ol{L_n}  &      L'_m &\cap L_n
\end{align*}
Although the complement of a regular language $L_n$ is usually taken with respect to $\Sigma^*$, where $\Sigma$ is the alphabet of $L_n$,
the list above requires that $\ol{L_n}$ denotes the complement of $L_n$ in a specific universal set $U$ which contains both $L'_m$ and $L_n$.
We wish for $U$ to be the set of all strings over some alphabet, and it is most natural to have $U = (\Sigma' \cup \Sigma)^*$, where $\Sigma'$ is the alphabet of $L'_m$.
Thus, contrary to its usual meaning, every use of complement in the list of operations above is taken with respect to $(\Sigma' \cup \Sigma)^*$.

We  study the complexities of four proper boolean operations only: union ($L'_m \cup L_n$), symmetric difference ($L'_m \oplus L_n$),
difference ($L'_m \setminus L_n$), and intersection ($L'_m \cap L_n$).
From these four it is generally a straightforward exercise to deduce the complexity of any other operation:
The complexity of $L_n \setminus L'_m$ is determined by symmetry with $L'_m \setminus L_n$, and
from De Morgan's laws we have
$\ol{L'_m \cup L_n} = \ol{L'_m} \cap \ol{L_n}$,
$\ol{L'_m \cap L_n} = \ol{L'_m} \cup \ol{L_n}$,
$\ol{L'_m \setminus L_n} = \ol{L'_m} \cup L_n$, 
$\ol{L_n \setminus L'_m} = L'_m \cup \ol{L_n}$, and
$\ol{L'_m \oplus L_n} = L'_m \oplus \ol{L_n}$.
As discussed in Terminology and Notation, $\kappa(L)$ and $\kappa(\ol{L})$ differ by at most 1 for any regular language $L$, and for any specific witness one can easily determine the discrepancy; for this reason we leave it as an exercise to verify that our witnesses meet the upper bounds for all ten proper operations based on the four operations that we address explicitly.

It turns out that the witnesses that we used for unrestricted product also work for unrestricted boolean operations.

\begin{theorem}[(Boolean Operations on Regular Languages)]
\label{thm:boolean}
For $m,n \ge 3$, let $L'_m$ (respectively, $L_n$) be a regular language with $m$ (respectively, $n$) quotients over an alphabet $\Sig'$, (respectively, $\Sig$). 
Then the complexity of union and symmetric difference is $(m+1)(n+1)$ 
and this bound is met by $L'_m(a,b,-,c)$ and $L_n(b,a,-,d)$; 
the complexity of difference is $mn+m$, and this bound is met by $L'_m(a,b,-,c)$ and
$L_n(b,a)$; the complexity of intersection is $mn$ and this bound is met by  
$L'_m(a,b)$ and $L_n(b,a)$.

\end{theorem}

\begin{proof}
Let $\cD'_m = ( Q'_m, \Sig', \delta', 0',F')$  and 
$\cD_n = (Q_n, \Sig,\delta, 0, F)$ be minimal DFAs for arbitrary regular languages $L'_m$ and $L_n$  with $m$ and $n$ quotients, respectively.
To calculate an upper bound for the boolean operations assume that $\Sig'\setminus \Sig $ and 
$\Sig \setminus \Sig' $ are non-empty; this assumption results in the largest upper bound. 
We  add an empty state $\emp'$ to $\cD'_m$ to send all transitions under the letters from 
$\Sig \setminus \Sig' $ to that state; thus we get an $(m+1)$-state DFA  $\cD'_{m,\emp'}$. Similarly, we add an empty state $\emp$  to $\cD_n$ to get $\cD_{n,\emp}$.
Now we have two DFAs over the same alphabet, and an ordinary problem of finding an upper bound for the boolean operations on two languages over the same alphabet, \emph{except that these languages both have empty quotients}. 
It is clear that $(m+1)(n+1)$ is an upper bound for all four operations, but it can be improved for difference and intersection.
Consider the direct product $\cP_{m,n}$ of $\cD'_{m,\emp'}$ and $\cD_{n,\emp}$. 

For difference, all $n+1$ states of $\cP_{m,n}$ that have the form $(\emp', q)$, where $q\in Q_n\cup \{\emp\}$ are empty. Hence the bound can be reduced by $n$ states to $mn+m+1$. 
However, the empty states can only be reached by words in $\Sig\setminus \Sig'$ and the alphabet of $L'_m\setminus L_n$ is a subset of $\Sig'$; hence  the bound is reduced futher to $mn+m$.

For intersection, all $n$ states $(\emp',q)$, $q\in Q_n$, and all $m$ states $(p',\emp)$, $p'\in Q'_m$, are equivalent to the empty state $(\emp',\emp)$, thus reducing the upper bound to
$mn+1$. Since the alphabet of $L'_m\cap L_n$ is a subset of $\Sig'\cap \Sig$, these empty states cannot be reached and the bound is reduced to $mn$.

To prove that the bounds are tight, we start with $\cD_n(a,b,c,d)$ of Definition~\ref{def:regular}.
For $m,n\ge 3$,  let $D'_m(a,b,-,c)$ be the dialect of $\cD'_m(a,b,c,d)$ where $c$ plays the role of $d$ and the alphabet is restricted to $\{a,b,c\}$, and let
$\cD_n(b,a,-,d)$ be the dialect of $\cD_n(a,b,c,d)$ in which $a$ and $b$ are permuted, and the alphabet is  restricted to $\{a,b,d\}$;
see Figure~\ref{fig:boolean}. 
\begin{figure}[ht]
\unitlength 7.5pt
\begin{center}\begin{picture}(37,19)(-3.5,2)
\gasset{Nh=2.5,Nw=5.6,Nmr=1.25,ELdist=0.4,loopdiam=1.5}
	{\small
\node(0')(-2,14){$0'$}\imark(0')
\node(1')(7,14){$1'$}
\node(2')(16,14){$2'$}
\node[Nframe=n](3dots')(25,14){$\dots$}
\node(m-1')(34,14){$(m-1)'$}\rmark(m-1')
\drawedge[curvedepth= 1.4,ELdist=-1.3](0',1'){$a,b$}
\drawedge[curvedepth= 1,ELdist=.3](1',0'){$b$}
\drawedge(1',2'){$a$}
\drawedge(2',3dots'){$a$}
\drawedge(3dots',m-1'){$a$}
\drawedge[curvedepth= -5.2,ELdist=-1](m-1',0'){$a$}
\drawloop(0'){$c$}
\drawloop(1'){$c$}
\drawloop(2'){$b,c$}
\drawloop(m-1'){$b,c$}

\gasset{Nh=2.5,Nw=5.6,Nmr=1.25,ELdist=0.4,loopdiam=1.5}

\node(0)(-2,7){0}\imark(0)
\node(1)(7,7){1}
\node(2)(16,7){2}
\node[Nframe=n](3dots)(25,7){$\dots$}
\node(n-1)(34,7){$n-1$}\rmark(n-1)
\drawloop(0){$d$}
\drawloop(1){$d$}
\drawloop(2){$a,d$}
\drawloop(n-1){$a,d$}
\drawedge[curvedepth= 1.2,ELdist=-1.3](0,1){$a,b$}
\drawedge[curvedepth= .8,ELdist=.25](1,0){$a$}
\drawedge(1,2){$b$}
\drawedge(2,3dots){$b$}
\drawedge(3dots,n-1){$b$}
\drawedge[curvedepth= 5.0,ELdist=-1.5](n-1,0){$b$}
}
\end{picture}\end{center}
\caption{Witnesses $D'_m(a,b,-,c)$ and $\cD_n(b,a,-,d)$ for boolean operations. }
\label{fig:boolean}
\end{figure}
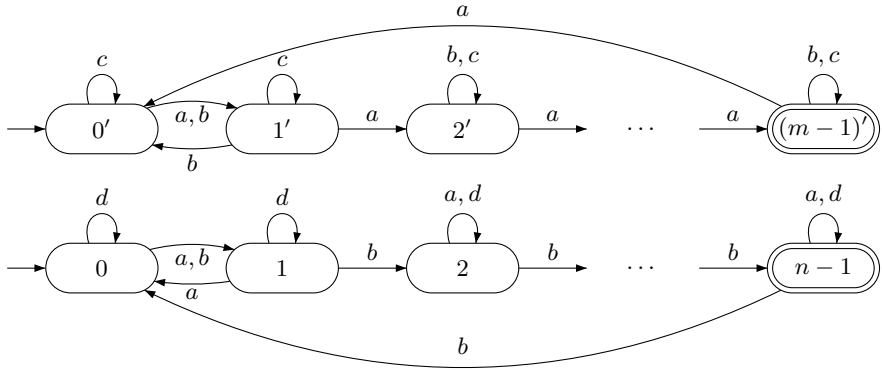
We complete the two DFAs by adding empty states, and then construct the direct product of the new DFAs as illustrated in Figure~\ref{fig:cross}.

If we restrict both DFAs to the alphabet $\{a,b\}$,  we have the usual problem of determining the complexity of two DFAs over the same alphabet. 
By \cite[Theorem 1]{BBMR14}, all $mn$ states of the form $(p',q)$, $p'\in Q'_m$, $q\in Q_n$, are reachable and pairwise distinguishable by words in $\{a,b\}^*$ for all proper boolean operations if $(m,n)\notin \{(3,4),(4,3),(4,4)\}$.
For our application, the three exceptional cases were verified by computation.

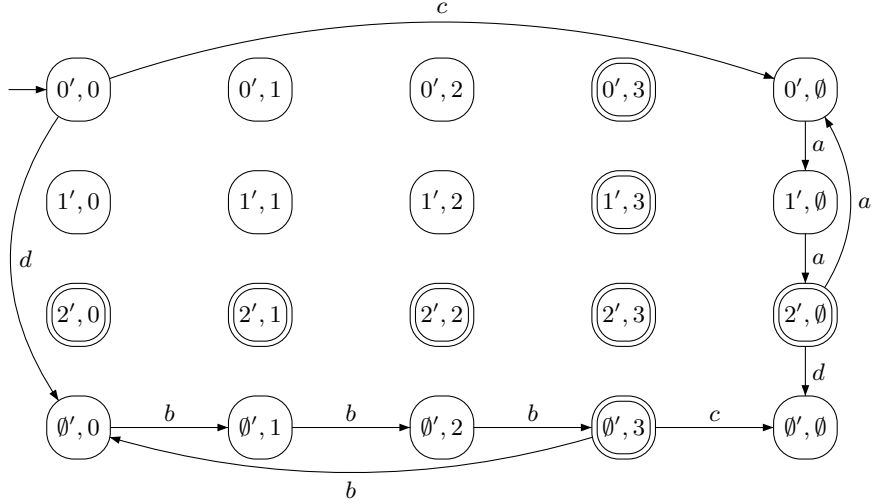
\begin{figure}[ht]
\unitlength 8.5pt
\begin{center}\begin{picture}(35,22)(0,-2)
\gasset{Nh=2.8,Nw=2.8,Nmr=1.2,ELdist=0.3,loopdiam=1.2}
	{\small
\node(0'0)(2,15){$0',0$}\imark(0'0)
\node(1'0)(2,10){$1',0$}
\node(2'0)(2,5){$2',0$}\rmark(2'0)
\node(3'0)(2,0){$\emp',0$}

\node(0'1)(10,15){$0',1$}
\node(1'1)(10,10){$1',1$}
\node(2'1)(10,5){$2',1$}\rmark(2'1)
\node(3'1)(10,0){$\emp',1$}

\node(0'2)(18,15){$0',2$}
\node(1'2)(18,10){$1',2$}
\node(2'2)(18,5){$2',2$}\rmark(2'2)
\node(3'2)(18,0){$\emp',2$}

\node(0'3)(26,15){$0',3$}\rmark(0'3)
\node(1'3)(26,10){$1',3$}\rmark(1'3)
\node(2'3)(26,5){$2',3$}\rmark(2'3)
\node(3'3)(26,0){$\emp',3$}\rmark(3'3)

\node(0'4)(34,15){$0',\emp$}
\node(1'4)(34,10){$1',\emp$}
\node(2'4)(34,5){$2',\emp$}\rmark(2'4)
\node(3'4)(34,0){$\emp',\emp$}
	
\drawedge(3'0,3'1){$b$}
\drawedge(3'1,3'2){$b$}
\drawedge(3'2,3'3){$b$}
\drawedge[curvedepth=2,ELdist=.4](3'3,3'0){$b$}

\drawedge(0'4,1'4){$a$}
\drawedge(1'4,2'4){$a$}
\drawedge[curvedepth=-2,ELdist=-.9](2'4,0'4){$a$}
\drawedge(3'3,3'4){$c$}
\drawedge(2'4,3'4){$d$}

\drawedge[curvedepth=-3,ELdist=.4](0'0,3'0){$d$}
\drawedge[curvedepth=3,ELdist=.4](0'0,0'4){$c$}
}
\end{picture}\end{center}
\caption{Direct product for union shown partially.}
\label{fig:cross}
\end{figure}

To prove that the remaining states are reachable, observe that 
$(0',0) \stackrel{d} {\longrightarrow} (\emp',0)$ and
$(\emp',0) \stackrel{b^q} {\longrightarrow} (\emp',q)$, for $q\in Q_n$.
Symmetrically, 
$(0',0) \stackrel{c} {\longrightarrow} (0',\emp)$ and
$(0',\emp) \stackrel{a^p} {\longrightarrow} (p',\emp)$, for $p'\in Q'_m$.
Finally, $(\emp',n-1) \stackrel{c} {\longrightarrow} (\emp',\emp)$, and  all $(m+1)(n+1)$ states of the direct product are reachable.

It remains to verify that the appropriate states are pairwise distinguishable.
From \cite[Theorem 1]{BBMR14}, we know that all states in $Q'_m\times Q_n$ are distinguishable.
Let $H= \{(\emp',q) \mid q\in Q_n  \}$, and
$V= \{ (p',\emp) \mid p'\in Q'_m \}$.
For the operations consider four cases:
\bd

\item[\bf Union]
The final states of $\cP_{m,n}$ are $\{((m-1)',q) \mid q\in Q_n\cup \{\emp\}  \}$,  and
$\{ (p',n-1) \mid p'\in Q'_m \cup \{\emp'\} \}$.
Every state in $V$ accepts a word with a $c$, whereas  no state in $H$ accepts such words.
Similarly, every state in $H$ accepts a word with a $d$, whereas  no state in $V$ accepts such words.
Every state in $Q'_m \times Q_n$ accepts a word with a $c$ and a word with    a $d$. State $(\emp',\emp)$ accepts no words at all.
Hence any two states chosen from different sets (the sets being $Q'_m\times Q_n$, $H$, $V$,  and $\{(\emp',\emp)\}$) are distinguishable.
States in $H$ are distinguishable by words in $b^*$ and those in $V$, by words in $a^*$.
Therefore all $mn+m+n+1$ states are pairwise distinguishable.

\item[\bf Symmetric Difference]
The final states here are all the final states for union except $( (m-1)',n-1 )$. The rest of the argument is the same as for union.

\item[\bf Difference]
Here the final states are $\{((m-1)', q) \mid q\neq n-1\}$.
The $n$ states of the form $(\emp',q)$, $q \in Q_n$, are now equivalent to the empty state $(\emp',\emp)$. 
The remaining states are non-empty as each accepts a word in $ca^*$.
The states of $V$ are pairwise distinguishable by words in $a^*$.
A state $(p', q) \in Q'_m \times Q_n$ is distinguished from $(r', \emp) \in V$ by $a^{m-1-r}$, unless $p' = r'$.
If $p' = r'$, they are distinguished by a word in $\{a, b\}^*$ that maps $(p', q)$ to $((m-1)', n-1)$, for this word must send $(r', \emp)$ to $((m-1)', \emp)$.
Hence we have $mn+m+1$ distinguishable states.
However, the alphabet of $L'_m \setminus L_n$ is $\{a,b,c\}$, and  the empty state can only be reached by $d$. Since this empty state is not needed, neither is $d$, and the final bound is $mn+m$; it is reached by $L'_m(a,b,-,c)$ and $L_n(b,a)$.

\item[\bf Intersection]
Here only $((m-1)', n-1 )$ is final and all states $(p', \emp)$, $p' \in Q'_m$, and $(\emp',q)$, $q\in Q_n$ are equivalent to $(\emp',\emp)$, leaving $mn+1$ distinguishable states.
 However, the alphabet of $L'_m \cap L_n$ is $\{a,b\}$, and so the empty state cannot be reached. This gives the final bound of $mn$ states, and this bound is met by $L'_m(a,b)$ and $L_n(b,a)$ as was already known in~\cite{Brz13}.\vspace*{-\baselineskip}
\ed
\end{proof}
\vspace*{-0.15cm}

\begin{remark}
 In the restricted case the  complexity of every one of the ten binary boolean functions in $mn$. 
In the unrestricted case one verifies that we have
$$\kappa(L'_m \cup L_n) = \kappa(\ol{L'_m} \cap \ol{L_n})=
\kappa(L'_m\oplus L_n)= \kappa(L'_m \oplus \ol{L_n})=(m+1)(n+1),$$
$$\kappa(\ol{L'_m} \cup L_n)= mn+m+1, 
\kappa(L'_m \cup \ol{L_n})= mn+n+1,$$
$$\kappa(L'_m\cap \ol{L_n}) =mn+m,
\kappa(\ol{L_m} \cap L_n) = mn +n, $$
$$\kappa(\ol{L'_m} \cup \ol{L_n})= mn+1,
\kappa(L'_m \cap L_n) =mn.$$

As before, complement is taken with respect to $(\Sigma' \cup \Sigma)^*$ for boolean operations.

\label{rem:boolean}
\end{remark}

\subsection{Most Complex Regular Languages}

We now update the result of~\cite{Brz13} to include the unrestricted case.

\begin{theorem}[(Most Complex Regular Languages)]
\label{thm:main}
For each $n\ge 3$, the DFA of Definition~\ref{def:regular} is minimal and its 
language $L_n(a,b,c,d)$ has complexity $n$.
The stream $(L_m(a,b,c,d) \mid m \ge 3)$  with dialect streams
$(L_n(a,b,-,c) \mid n \ge 3)$ and $(L_n(b,a,-,d) \mid n \ge 3)$
is most complex in the class of regular languages.
In particular, it meets all the complexity bounds below, which are maximal for regular languages.
In several cases the bounds can be met with a reduced alphabet.
\begin{enumerate}
\item
The syntactic semigroup of $L_n(a,b,c)$ has cardinality $n^n$.  
\item
Each quotient of $L_n(a)$ has complexity $n$.
\item
The reverse of $L_n(a,b,c)$ has complexity $2^n$, and $L_n(a,b,c)$ has $2^n$ atoms.
\item
For each atom $A_S$ of $L_n(a,b,c)$, the complexity $\kappa(A_S)$ satisfies:
\begin{equation*}
	\kappa(A_S) =
	\begin{cases}
		2^n-1, 			& \text{if $S\in \{\emp,Q_n\}$;}\\
		1+ \sum_{x=1}^{|S|}\sum_{y=1}^{n-|S|} \binom{n}{x}\binom{n-x}{y},
		 			& \text{if $\emp \subsetneq S \subsetneq Q_n$.}
		\end{cases}
\end{equation*}
\item
The star of $L_n(a,b)$ has complexity $2^{n-1}+2^{n-2}$.
\item Product
	\be
	\item { Restricted case:}\\
	The product $L'_m(a,b,c) L_n(a,b,c)$ has complexity $m2^n-2^{n-1}$.
	\item Unrestricted case:\\
	The product $L'_m(a,b,-,c) L_n(b,a,-,d)$ has complexity $m2^n+2^{n-1}$.
	\ee
\item Boolean operations
	\be
	\item
	Restricted case:\\
	The complexity of $L'_m(a,b) \circ L_n(b,a)$ is $mn$.
	\item
	Unrestricted case:\\
	The complexity of $L'_m(a,b,-,c) \circ L_n(b,a,-,d)$ 
	is $(m+1)(n+1)$ if $\circ\in \{\cup,\oplus\}$,
	that of  $L'_m(a,b,-,c) \setminus L_n(b,a)$ is 
 	$mn+m$, and 
 	that of $L'_m(a,b) \cap L_n(b,a)$ is $mn$. 
 	\ee

\end{enumerate}
\end{theorem}
\begin{proof}
The proofs for the restricted case can be found in~\cite{Brz13}, and
the claims for the unrestricted case were proved in the present paper, in Theorems~\ref{thm:product} and~\ref{thm:boolean}.
\end{proof}

\begin{proposition}[(Marek Szyku{\l}a, personal communication)]
At least four letters are required for a most complex regular language. In particular, four letters are needed for union: two letters  to reach all pairs of states in $Q'_m \times Q_n$, one in $\Sigma' \setminus \Sigma$ for pairs $(p',\emp)$ with $p'\in Q'_m$, and one in $\Sigma \setminus \Sigma'$ for pairs $(\emp',q)$ with $q \in  Q_n$.
\end{proposition}

\begin{center} $\ast \ast \ast$ \end{center}

A non-empty language $L$ is a right ideal (left ideal, two-sided ideal, respectively) if $L=L\Sig^*$
($L=\Sig^*L$, $L=\Sig^*L\Sig^*$, respectively).
Ideals are fundamental objects in semigroup theory.
We study regular ideals only.

Ideals appear in the area of pattern matching~\cite{CrHa90}.
 For this application, a \emph{text} is represented by a word $w$ over some alphabet $\Sig$. 
A \emph{pattern} can be an arbitrary language $L$ over $\Sig$ described by a regular expression.
An occurrence of a pattern represented by $L$ in text $w$ is a triple $(u,x,v)$ such that $w=uxv$ and $x$ is in~$L$.
Searching text $w$ for words in $L$ is equivalent to looking for prefixes of $w$ that belong to the language $\Sig^*L$, which is the left ideal generated by $L$.

Algorithms such as that of Aho and Corasick~\cite{AhCo75} can be used to determine all possible occurrences of words from a finite set $L$ in a given input $w$.
For example, in  a  Unix-style editor such as \emph{sed}, one can find all the words ending in $x$ (that~is,~all the words 
of the left ideal $\Sig^*x$) that occur in $w$; all the words beginning with $x$ (that is, all the words of the right ideal $x\Sig^*$) that occur in $w$; and
all the words that have $x$ as a factor (that~is,~all the words of the two-sided ideal $\Sig^*x\Sig^*$)  that occur in $w$.

The complexities of restricted basic operations on ideals were studied in~\cite{BJL13}. 
The sizes of transition semigroups of minimal DFAs accepting ideal languages were determined in~\cite{BrSz14a,BrYe11}. Atoms of ideals were analyzed in~\cite{BrDa15}.
Most complex right ideals for restricted operations were studied in~\cite{BrDa14}, and left and two-sided ideals were added in~\cite{BDL15}.
In this paper we add the results for unrestricted binary operations.

\section{Right Ideals}\
A stream of right ideals that have the largest syntactic semigroups was introduced in~\cite{BrYe11}. A different stream was used in~\cite{BrDa14} and shown to be
most complex for restricted operations; it was also studied in~\cite{BrDa15,BDL15}. Here we modify the stream of~\cite{BrDa14} by adding an input $e$ that induces the identity transformation.  We find the unrestricted state complexity of product and boolean operations of this stream together with some of its permutational dialects.

\begin{definition}
\label{def:RWit}
For $n\ge 3$, let $\cD_n=\cD_n(a,b,c,d,e)=(Q_n,\Sig,\delta_n, 1, \{n\})$, where 
$\Sig=\{a,b,c,d,e\}$, 
and $\delta_n$ is defined by 
$a\colon (0,\dots,n-2)$,
$b\colon(1,\ldots,n-2)$,
${c\colon(n-2\rightarrow 0)}$,
${d\colon(n-2\rightarrow n-1)}$,
and $e \colon \mathbbm{1}$.
Let $L_n=L_n(a,b,c,d,e)$ be the language accepted by~$\cD_n$.
The structure of $\cD_n$ is shown in Figure~\ref{fig:RWit}. 
\end{definition}

\begin{figure}[h]
\unitlength 10.8pt
\begin{center}\begin{picture}(31,7.5)(-2,0)
\gasset{Nh=2,Nw=2.8,Nmr=1.2,ELdist=0.3,loopdiam=1.2}
{\small
\node(1)(2,4){$0$}\imark(1)
\node(2)(6,4){$1$}
\node(3)(10,4){$2$}
\node[Nframe=n](qdots)(14,4){$\dots$}
\node(n-2)(18,4){{\small $n-3$}}
\node(n-1)(22,4){{\small $n-2$}}
\node(n)(26,4){{\small $n-1$}}\rmark(n)

\drawedge(1,2){$a$}
\drawedge(2,3){$a,b$}
\drawedge(3,qdots){$a,b$}
\drawedge(qdots,n-2){$a,b$}
\drawedge(n-2,n-1){$a,b$}
\drawedge(n-1,n){$d$}
\drawedge[curvedepth=2.1](n-1,2){$b$}
\drawedge[curvedepth=3.8](n-1,1){$a,c$}

\drawloop(1){$b,c,d,e$}
\drawloop(2){$c,d,e$}
\drawloop(3){$c,d,e$}
\drawloop(n-2){$c,d,e$}
\drawloop(n-1){$e$}
\drawloop(n){$a,b,c,d,e$}
}
\end{picture}\end{center}
\caption{Minimal DFA $\cD_n(a,b,c,d,e)$  of Definition~\ref{def:RWit}.}
\label{fig:RWit}
\end{figure}
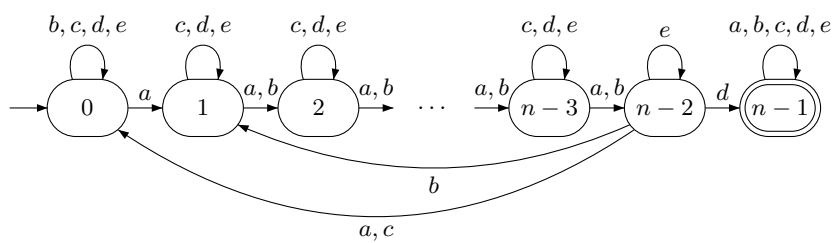

\subsection{Product of Right Ideals}

\begin{theorem}[Product of Right Ideals)]
\label{thm:RWitProduct}
For $m,n \ge 3$, let $L'_m$ (respectively, $L_n$) be an arbitrary right ideal with $m$ (respectively, $n$) quotients over an alphabet $\Sig'$, (respectively, $\Sig$). 
Then $\kappa(L'_mL_n) \le m + 2^{n-2} + 2^{n-1} + 1$, and this bound is met by 
the right ideals $L'_m(a,b,-,d, e)$ and $L_n(a,b,-,d,c)$ of Definition~\ref{def:RWit}.
\end{theorem}
\begin{proof}
Let $\mathcal{D}'{_m}= (Q'_m,\Sigma', \delta', 0', \{(m-1)'\})$ and $\mathcal{D}{_n}= (Q_n,\Sigma, \delta, 0, \{n-1\})$ be minimal DFAs of arbitrary right ideals $L'_m$ and $L_n$, respectively.
We use the standard  construction of the NFA for the product $L'_mL_n$.
We bound the complexity of the product by counting the reachable and distinguishable sets in the subset construction on this NFA.
The $m-1$  sets $\{p'\}$, where $p'\in Q'_m$ is non-final,   as as well as $\{(m-1)',0\}$ may be reachable.
From  $\{(m-1)',0\}$ we may reach all $2^{n-1}$ sets $\{(m-1)', 0\} \cup S$ for $S \subseteq Q_n \setminus \{0\}$.
The $2^{n}$ sets $S \subseteq Q_n$ may also be reachable by using a letter in $\Sigma \setminus \Sigma'$.
So far, there are  at most  $m-1 + 2^{n-1} + 2^n$ reachable sets. 
However, in the DFA obtained by the subset construction from the NFA, states $\{n-1\} \cup S$ for $S \subseteq Q_{n-1}$ are all equivalent,
since they accept all words in $\Sigma^*$ and go to $\emptyset$ by letters in $\Sigma' \setminus \Sigma$.
Similarly, states $\{(m-1)', 0, n-1\} \cup S$ are all equivalent because they all accept $\Sig^*$ and go to $\{(m-1)',0\}$ by letters in $\Sigma' \setminus \Sigma$. This leaves at most  $m + 2^{n-2} + 2^{n-1} + 1$ distinguishable states.

\begin{figure}[th]
\unitlength 8.5pt
\begin{center}\begin{picture}(38,22)(-4,0)
\gasset{Nh=2.2,Nw=5.,Nmr=1.25,ELdist=0.4,loopdiam=1.5}
{\small
\node(0')(-4,14){$0'$}\imark(0')
\drawloop(0'){$b,d,e$}
\node(1')(3,14){$1'$}
\drawloop(1'){$d,e$}
\node(2')(10,14){$2'$}
\drawloop(2'){$d,e$}
\node[Nframe=n](3dots')(17,14){$\dots$}
\node(m-2')(24,14){$(m-2)'$}
\drawloop(m-2'){$e$}
\node(m-1')(24,10){$(m-1)'$}
\drawloop[loopangle=190,ELdist=.3](m-1'){$a,b,d,e$}

\drawedge(0',1'){$a$}
\drawedge(1',2'){$a,b$}
\drawedge(2',3dots'){$a,b$}
\drawedge(3dots',m-2'){$a,b$}
\drawedge[curvedepth= -5.7,ELside=r](m-2',0'){$a$}
\drawedge[curvedepth= -3,ELside=r](m-2',1'){$b$}
\drawedge(m-2',m-1'){$d$}

\node(0)(3,6){0}
\drawloop[loopangle=270,ELdist=.3](0){$b,c,d$}
\node(1)(10,6){1}
\drawloop[loopangle=270,ELdist=.3](1){$c,d$}
\node(2)(17,6){2}
\drawloop[loopangle=270,ELdist=.3](2){$c,d$}
\node[Nframe=n](3dots)(24,6){$\dots$}
\node(n-2)(31,6){$n-2$}
\drawloop[loopangle=270,ELdist=.3](n-2){$c$}
\node(n-1)(31,10){$n-1$}\rmark(n-1)
\drawloop(n-1){$a,b,c,d$}
\drawedge(0,1){$a$}
\drawedge(1,2){$a,b$}
\drawedge(2,3dots){$a,b$}
\drawedge(3dots,n-2){$a,b$}
\drawedge[curvedepth= 3](n-2,1){$b$}
\drawedge[curvedepth= 5.5](n-2,0){$a$}
\drawedge(n-2,n-1){$d$}
\drawedge[curvedepth= -3,ELside=r](m-1',0){$\varepsilon$}
}
\end{picture}\end{center}
\caption{An NFA  for product of right ideals  $L'_m(a,b,-,d,e)$ and $L_n(a,b,-,d,c)$.}
\label{fig:RIdealProduct}
\end{figure}
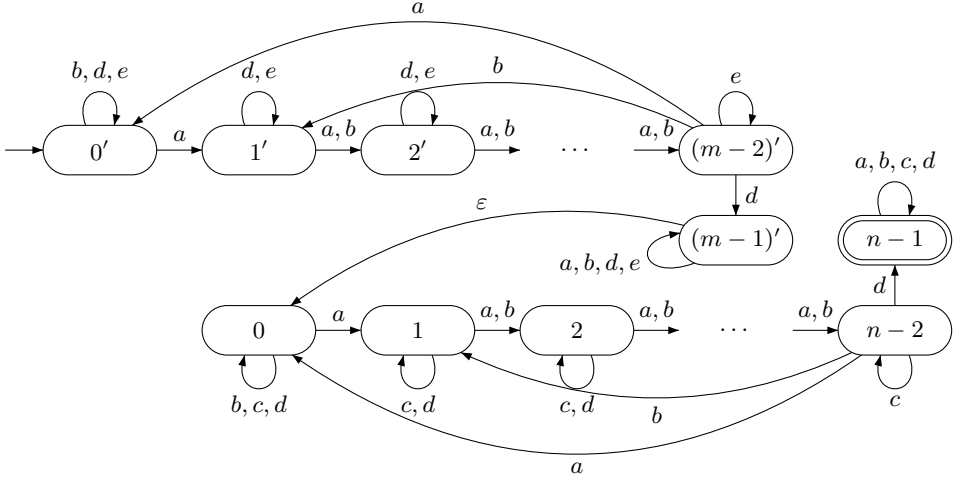

 To show that this bound can be met, consider two dialects of the DFA of Definition~\ref{def:RWit} shown in Figure~\ref{fig:RIdealProduct}.
Here $\Sigma' = \{a,b,d,e\}$ and $\Sigma = \{a,b,c,d\}$.
Set $\{p'\}$ for $p' \in Q'_{m-1}$ is reached by $a^{p}$, and $\{(m-1)', 0\}$ is reached by $a^{m-2}d$.
Any set $\{(m-1)', 0, q_1, q_2, \dots, q_k\}$ with $0 < q_1 < q_2 < \dots < q_k < n-1$ and $k \ge 1$ is reached from 
$\{(m-1)', 0, q_2-q_1 -1, q_3 -q_1-1, \dots, q_k-q_1-1\}$ by $ab^{q_1-1}$;
hence all sets $\{(m-1)', 0\} \cup S$ where $S \subseteq Q_{n-1} \setminus \{0\}$ are reachable.
Set $S = \{q_1, q_2, \dots, q_k\}$ with $0 \le q_1 < q_2 < \dots < q_k < n-1$ and $k \ge 1$  is reachable from 
$\{(m-1)', 0, q_2 - q_1, \dots, q_k - q_1\}$ by $ca^{q_1}$, and $\emptyset$ is reached from $\{0\}$ by $c$; hence all sets $S \subseteq Q_{n-1}$ are reachable.
Sets containing $n-1$ are easily reached from these sets using $a$ and $d$. However, in the DFA obtained from the NFA, the states $S \subseteq Q_n$ that contain $n-1$ all accept $\Sigma^*$ and are sent to the empty state by $e$; hence they are all equivalent.
Similarly, the states $\{(m-1)', 0\} \cup S$ that contain $n-1$ all accept $\{a,b,c,d\}^*$ and are sent to $\{(m-1)', 0\}$ by $e$; hence they are also equivalent.

The remaining states are pairwise distinguishable.
States $\{p'\}$ and $\{q'\}$ with $0\le p<q \le m-2$ are distinguished by $a^{m-2-q}da^{n-2}d$, and 
 $\{p'\}$ is distinguished from $\{(m-1)',0\} \cup S$ 
 or from $S$, where $\emptyset \subsetneq S\subseteq Q_n$, by $(ad)^{n-1}$.
Two states $\{(m-1)',0\} \cup S$ and
$\{(m-1)',0\} \cup T$ with $q\in S\oplus T$ are distinguished by $a^{n-2-q}d$, as are 
two states $S$ and $T$ with $q\in S\oplus T$.
A set $\{(m-1)',0\} \cup S$ is distinguishable from $T$, where $S,T \subseteq Q_n$, by $ea^{n-2}d$.
Every state is distinguishable from $\emptyset$ by a word in $\{a,d\}^*$.
Thus all $m+2^{n-2}+2^{n-1}+1$ states are pairwise distinguishable.
\end{proof}

\subsection{Boolean Operations on Right Ideals}

\begin{theorem}[(Boolean Operations on Right Ideals)]
\label{thm:RWitBoolean}
For $m,n \ge 3$, the unrestricted complexities of boolean operations on right ideals are the same as those for arbitrary regular languages.
In particular, 
 the right ideals $L'_m(a,b,-,d,e)$ and $L_n(e,c,-,d,a)$ of Definition~\ref{def:RWit}
meet the bound $(m+1)(n+1)$ for union and symmetric difference,
$L'_m(a,b,-,d,e)$ and
$L_n(e,-,-,d,a)$ meet the bound $mn+m$ for difference, and $L'_m(a,-,-,d,e)$ and $L_n(e,-,-,d,a)$ meet the bound $mn$ for  intersection.
\end{theorem}

\begin{proof}
We show that the bounds of Theorem~\ref{thm:boolean} for arbitrary regular languages are met by DFAs  $\mathcal{D}'_m(a,b,-,d,e)$ and $\mathcal{D}_n(e,c,-,d,a)$ of the right ideals of Definition~\ref{def:RWit}.

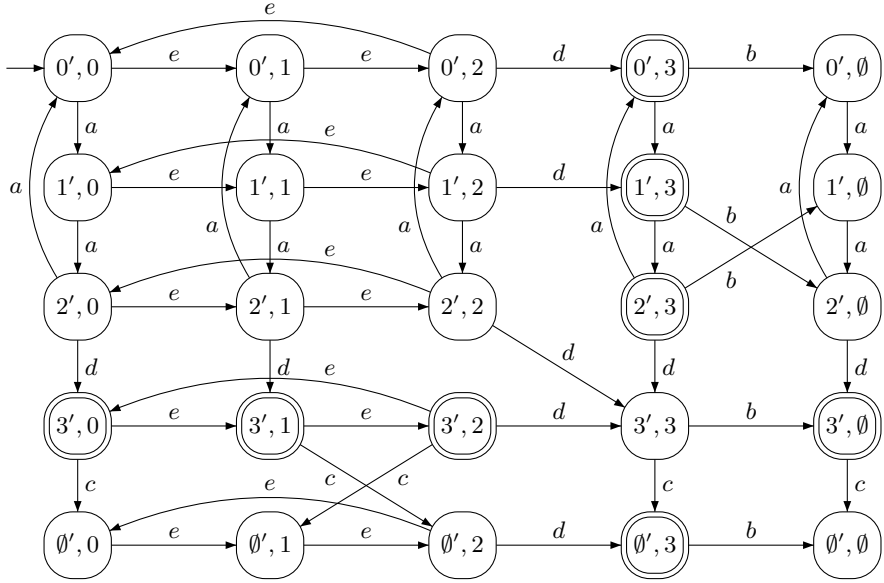
\begin{figure}[th]
\unitlength 9pt
\begin{center}\begin{picture}(35,23)(0,-6)
\gasset{Nh=2.8,Nw=2.8,Nmr=1.2,ELdist=0.3,loopdiam=1.2}
	{\small
\node(0'0)(2,15){$0',0$}\imark(0'0)
\node(1'0)(2,10){$1',0$}
\node(2'0)(2,5){$2',0$}
\node(3'0)(2,0){$3',0$}\rmark(3'0)
\node(4'0)(2,-5){$\emptyset',0$}

\node(0'1)(10,15){$0',1$}
\node(1'1)(10,10){$1',1$}
\node(2'1)(10,5){$2',1$}
\node(3'1)(10,0){$3',1$}\rmark(3'1)
\node(4'1)(10,-5){$\emptyset',1$}

\node(0'2)(18,15){$0',2$}
\node(1'2)(18,10){$1',2$}
\node(2'2)(18,5){$2',2$}
\node(3'2)(18,0){$3',2$}\rmark(3'2)
\node(4'2)(18,-5){$\emptyset',2$}

\node(0'3)(26,15){$0',3$}\rmark(0'3)
\node(1'3)(26,10){$1',3$}\rmark(1'3)
\node(2'3)(26,5){$2',3$}\rmark(2'3)
\node(3'3)(26,0){$3',3$}
\node(4'3)(26,-5){$\emptyset',3$}\rmark(4'3)

\node(0'4)(34,15){$0',\emptyset$}
\node(1'4)(34,10){$1',\emptyset$}
\node(2'4)(34,5){$2',\emptyset$}
\node(3'4)(34,0){$3',\emptyset$}\rmark(3'4)
\node(4'4)(34,-5){$\emptyset',\emptyset$}
	
\drawedge(0'0,1'0){$a$}
\drawedge(1'0,2'0){$a$}
\drawedge[ELside=l,curvedepth=2](2'0,0'0){$a$}
\drawedge(0'1,1'1){$a$}
\drawedge(1'1,2'1){$a$}
\drawedge[ELside=l,ELpos=35,curvedepth=2](2'1,0'1){$a$}
\drawedge(0'2,1'2){$a$}
\drawedge(1'2,2'2){$a$}
\drawedge[ELside=l,ELpos=35,curvedepth=2](2'2,0'2){$a$}
\drawedge(0'3,1'3){$a$}
\drawedge(1'3,2'3){$a$}
\drawedge[ELside=l,ELpos=35,curvedepth=2](2'3,0'3){$a$}
\drawedge(0'4,1'4){$a$}
\drawedge(1'4,2'4){$a$}
\drawedge[ELside=l,curvedepth=2](2'4,0'4){$a$}

\drawedge(0'0,0'1){$e$}
\drawedge(0'1,0'2){$e$}
\drawedge[ELside=r,curvedepth=-2](0'2,0'0){$e$}
\drawedge(1'0,1'1){$e$}
\drawedge(1'1,1'2){$e$}
\drawedge[ELside=r,ELpos=35,curvedepth=-2](1'2,1'0){$e$}
\drawedge(2'0,2'1){$e$}
\drawedge(2'1,2'2){$e$}
\drawedge[ELside=r,ELpos=35,curvedepth=-2](2'2,2'0){$e$}
\drawedge(3'0,3'1){$e$}
\drawedge(3'1,3'2){$e$}
\drawedge[ELside=r,ELpos=35,curvedepth=-2](3'2,3'0){$e$}
\drawedge(4'0,4'1){$e$}
\drawedge(4'1,4'2){$e$}
\drawedge[ELside=r,curvedepth=-2](4'2,4'0){$e$}

\drawedge(2'0,3'0){$d$}
\drawedge(2'1,3'1){$d$}
\drawedge(2'2,3'3){$d$}
\drawedge(2'3,3'3){$d$}
\drawedge(2'4,3'4){$d$}

\drawedge(0'2,0'3){$d$}
\drawedge(1'2,1'3){$d$}
\drawedge(3'2,3'3){$d$}
\drawedge(4'2,4'3){$d$}

\drawedge(0'3,0'4){$b$}
\drawedge[ELpos=35](1'3,2'4){$b$}
\drawedge[ELside=r,ELpos=35](2'3,1'4){$b$}
\drawedge(3'3,3'4){$b$}
\drawedge(4'3,4'4){$b$}

\drawedge(3'0,4'0){$c$}
\drawedge[ELside=r,ELpos=35](3'1,4'2){$c$}
\drawedge[ELpos=35](3'2,4'1){$c$}
\drawedge(3'3,4'3){$c$}
\drawedge(3'4,4'4){$c$}
}
\end{picture}\end{center}
\caption{Partial illustration of the direct product for  $L'_4(a,b,-,d,e) \oplus L_4(e,c,-,d,a)$.}
\label{fig:idealcross}
\end{figure}

To compute $L'_m(a,b,-,d,e) \circ L_n(e,c,-,d,a)$, where $\circ$ is a boolean operation, add an empty state $\emptyset'$ to $\mathcal{D}'_m(a,b,-,d,e)$, and send all the transitions from any state of $Q'_m$ under $c$ to $\emptyset'$.
Similarly, add an empty state $\emptyset$ to  $\mathcal{D}_n(e,c,-,d,a)$  together with appropriate transitions; now the alphabets of the resulting DFAs are the same.
The direct product of $L'_m(a,b,-,d,e)$ and $L_n(e,c,-,d,a)$ is illustrated in Figure~\ref{fig:idealcross} for $m=n=4$.

We first check that all $(m+1)(n+1)$ states of the direct product are reachable.
State $(p',q) \in Q'_{m-1} \times Q_{n-1}$ is reached by $a^pe^q$.
State $((m-1)', q)$ for $q \in Q_{n-1}$ is reached from $((m-2)', 0)$ by $de^q$, state $(p', n-1)$ for $p' \in Q'_{m-1}$ is similarly reached from $(0', n-2)$, and $((m-1)',n-1)$ is reached from $((m-2)', n-2)$ by $d$.
Hence the states of $Q'_m \times Q_n$ are reachable; the remaining states involve $\emptyset'$ or $\emptyset$ and are reached from states of $Q'_m \times Q_n$ using $b$ and $c$ as seen in Figure~\ref{fig:idealcross}.

We now check distinguishability, which depends on the final states of the DFA.
\bd
\item[\textbf{Union}]
The final states of the direct product for union are $((m-1)', q)$ for $q \in Q_n \cup \{\emptyset\}$ and $(p', n-1)$ for $p' \in Q'_m \cup \{\emptyset'\}$.
States that differ in the first coordinate are distinguished by words in $ba^*d$, and states that differ in the second coordinate are distinguished by words in $ce^*d$;
hence all $(m+1)(n+1)$ states are distinguishable.

\item[\textbf{Symmetric Difference}]
The final states for symmetric difference are the same as those for union, except that $((m-1)', n-1)$ is non-final.
All states are pairwise distinguishable as in union.

\item[\textbf{Difference}]
The final states for difference are $((m-1)', q)$ for $q \not= n-1$.
The alphabet of $L'_m(a,b,-,d,e) \setminus L_n(e,c,-,d,a)$ is $\{a,b,d,e\}$; hence we can omit $c$ and delete all states $(\emptyset', q)$ and be left with a DFA recognizing the same language.
The remaining states are distinguished by words in $ba^*d$ if they differ in the first coordinate or by words in $a^* de^*d$ if they differ in the second coordinate.
The bound for difference is also met by $L'_m(a,b,-,d,e) \setminus L_n(e,-,-,d,a)$.

\item[\textbf{Intersection}]
For intersection, the only final state is $((m-1)', n-1)$.
The alphabet of $L'_m(a,b,-,d,e) \cap L_n(e,c,-,d,a)$ is $\{a,d,e\}$; hence we can omit $b$ and $c$ and delete all states $(p', \emptyset)$ and $(\emptyset', q)$.
The remaining $mn$ states are pairwise distinguishable by words in $e^*da^*d$.
The bound for intersection is also met by $L'_m(a,-,-,d,e) \setminus L_n(e,-,-,d,a)$.
\ed
One verifies that the complexities of all ten boolean functions on right ideals are as given in Remark~\ref{rem:boolean}.
\end{proof}

\subsection{Most Complex Right Ideals}

\begin{theorem}[(Most Complex Right Ideals)]
\label{thm:rightideals}
For each $n\ge 3$, the DFA of Definition~\ref{def:RWit} is minimal and  $L_n(a,b,c,d,e)$ is a right ideal of complexity $n$.
The stream $(L_n(a,b,c,d,e) \mid n \ge 3)$  with some dialect streams 
is most complex in the class of  right ideals.
In particular, it meets all the  bounds  below.
In several cases the bounds can be met with a reduced alphabet.
\begin{enumerate}
\item
The syntactic semigroup of $L_n(a,b,c,d)$ has cardinality $n^{n-1}$.  
\item
The quotients of $L_n(a,-,-,d)$ have complexity $n$ except for the final quotient which has complexity 1.
\item
The reverse of $L_n(a,-,-,d)$ has complexity $2^{n-1}$, and $L_n(a,-,-,d)$ has $2^{n-1}$ atoms.
\item
Each atom $A_S$\ of $L_n(a,b,c,d)$ has maximal complexity:
\begin{equation*}
	\kappa(A_S) =
	\begin{cases}
		2^{n-1}, 			& \text{if $S=Q_n$;}\\
		1 + \sum_{x=1}^{|S|}\sum_{y=1}^{n-|S|}\binom{n-1}{x-1}\binom{n-x}{y},
		 			& \text{if $\emptyset \subsetneq S \subsetneq Q_n$.}
		\end{cases}
\end{equation*}
\item
The star of $L_n(a,-,-,d)$ has complexity $n+1$.
\item Product
\begin{enumerate}
\item 
	Restricted case:\\
$L'_m(a,b,-,d) L_n(a,b,-,d)$ has complexity $m+2^{n-2}$.
\item
	Unrestricted case:\\
 $L'_m(a,b,-,d,e) L_n(a,b,-,d,c)$ has complexity $m+2^{n-1}+2^{n-2}+1$.
\end{enumerate}
\item Boolean operations
\begin{enumerate}
\item
	Restricted case:\\
The complexity of $\circ$
is $mn$ if $\circ\in \{\cap,\oplus\}$, 
 $mn-(m-1)$ if $\circ=\setminus$, and $mn-(m+n-2)$ if $\circ=\cup$, and 
these bounds are met by $L'_m(a,b,-,d) \circ L_n(b,a,-,d)$.
\item
	Unrestricted case:\\
The complexity of $L'_m(a,b,-,d,e) \circ L_n(e,c,-,d,a)$ is the same as  for arbitrary regular languages:
$(m+1)(n+1)$ if $\circ\in \{\cup,\oplus\}$, 
$mn+m$ if $\circ=\setminus$, and $mn$ if $\circ=\cap$.
The bound for difference is also met by $L'_m(a,b,-,d,e) \setminus L_n(e,-,-,d,a)$ and the bound for intersection  by $L'_m(a,-,-,d,e) \cap L_n(e,-,-,d,a)$.

\end{enumerate}
\end{enumerate}
\label{thm:ideal}
\end{theorem}
\begin{proof}
The restricted complexity results were proved in~\cite{BrDa14}, and the unrestricted results for product and boolean operations were proved in Theorems~\ref{thm:RWitProduct} and~\ref{thm:RWitBoolean}.
\end{proof}

\section{Left Ideals}

The following stream of left ideals was defined in~\cite{BrYe11}, where it was conjectured that its DFAs have maximal transition semigroups. This conjecture was proved in~\cite{BrSz14a}, and  it was shown in~\cite{BDL15} that this stream is most complex for restricted operations.
We prove that it is also most complex in the unrestricted case.

\begin{definition}
\label{def:leftideal}
For $n\ge 4$, let $\cD_n=\cD_n(a,b,c,d,e)=(Q_n,\Sig,\delta_n, 0, \{n-1\})$, where 
$\Sig=\{a,b,c,d,e\}$,
and $\delta_n$ is defined by  transformations
$a\colon (1,\dots,n-1)$,
$b\colon(1,2)$,
${c\colon(n-1 \to 1)}$,
${d\colon(n-1\to 1)}$, 
$e\colon (Q_n\to 1)$.
Let $L_n=L_n(a,b,c,d,e)$ be the language accepted by~$\cD_n$.
The structure of  $\cD_n(a,b,c,d,e)$ is shown in Figure~\ref{fig:LeftIdealWit}. 
\end{definition}

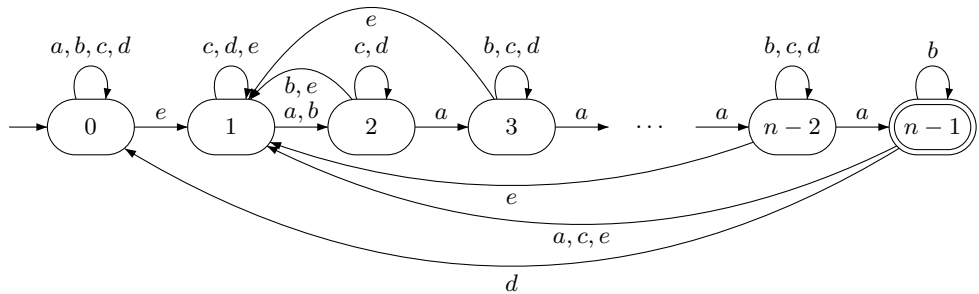
\begin{figure}[h]
\unitlength 10.5pt
\begin{center}\begin{picture}(31,9)(1,3)
\gasset{Nh=2,Nw=3.1,Nmr=1.2,ELdist=0.3,loopdiam=1.2}
{\small
\node(0)(2,8){$0$}\imark(0)
\node(1)(7,8){$1$}
\node(2)(12,8){$2$}
\node(3)(17,8){$3$}
\node[Nframe=n](qdots)(22,8){$\dots$}
{\scriptsize
\node(n-2)(27,8){{\small $n-2$}}
\node(n-1)(32,8){{\small $n-1$}}\rmark(n-1)
}
\drawedge(0,1){$e$}
\drawedge[ELdist=.15](1,2){$a,b$}
\drawedge(2,3){$a$}
\drawedge(3,qdots){$a$}
\drawedge(qdots,n-2){$a$}
\drawedge(n-2,n-1){$a$}
\drawloop(0){$a,b,c,d$}
\drawloop(1){$c,d,e$}
\drawloop(2){$c,d$}
\drawloop(3){$b,c,d$}
\drawloop(n-2){$b,c,d$}
\drawloop(n-1){$b$}
\drawedge[curvedepth=-2.1](2,1){$b,e$}
\drawedge[curvedepth=-4.3](3,1){$e$}
\drawedge[curvedepth=2.1](n-2,1){$e$}
\drawedge[curvedepth=3.5](n-1,1){$a,c,e$}
\drawedge[curvedepth=5](n-1,0){$d$}
}
\end{picture}\end{center}
\caption{Minimal DFA of left ideal of  Definition~\ref{def:leftideal}.}
\label{fig:LeftIdealWit}
\end{figure}

\subsection{Product of Left Ideals}

\begin{theorem}[(Product of Left Ideals)]
\label{thm:leftidealproduct}
For $m,n \ge 4$, let $L'_m$ (respectively, $L_n$) be an arbitrary left ideal with $m$ (respectively, $n$) quotients over an alphabet $\Sig'$, (respectively, $\Sig$). 
Then $\kappa(L'_m L_n) \le mn+m+n$, and this bound is met by $L'_m(a,b,-,d,e)$ and  $L_n(a,d,c,-,e)$ of Definition~\ref{def:leftideal}.
\end{theorem}
\begin{proof}
We first derive an upper bound on the complexity of $L'_mL_n$, where $L'_m$ is any regular language over $\Sigma'$ and $L_n$ is any left ideal over $\Sigma$. 
If $\Sig'\setminus\Sig$ and $\Sig\setminus \Sig'$ are non-empty, we consider $L'_m$ and $L_n$ to be over $\Sig' \cup\Sig$ by adding empty quotients.

We count the number of distinct quotients of $L'_mL_n$.
Consider $w^{-1}(L'_mL_n)$ for a word $w \in (\Sigma' \cup \Sigma)^*$. We decompose $w$ as $u_1u_2 \dots u_kv$ where $u_1\dots u_i \in L'_m$ for $i=1 ,\dots, k$ and no other prefix of $w$ is in $L'_m$. Then $wx \in L'_mL_n$ if and only if $u_{i+1} \dots u_kvx \in L_n$ for some $1 \le i \le k$, or $x=yz$ such that $wy \in L'_m$ and $z \in L_n$; hence 
$$w^{-1}(L'_mL_n) = \bigcup_{i=1}^k (u_{i+1}\dots u_kv)^{-1}L_n \cup (w^{-1}L'_m)L_n.$$

Notice that $v^{-1}L_n \subseteq (u_kv)^{-1}L_n \subseteq  (u_{k-1}u_kv)^{-1}L_n \subseteq \dots \subseteq (u_{i+1}\cdots u_kv)^{-1}L_n$ as long as $u_{i+1}\cdots u_kv \in \Sigma^*$, since $L_n$ is a left ideal with respect to $\Sigma$.
Alternatively if $u_{i+1}\cdots u_kv$ contains a letter from $\Sigma' \setminus \Sigma$ then $(u_{i+1}\cdots u_kv)^{-1}L_n = \emptyset$.
Thus, either $\bigcup_{i=1}^k (u_{i+1}\cdots u_kv)^{-1}L_n = \emptyset$, or there exists a minimal $j \in \{1, \cdots, k\}$ such that $u_{j+1}\cdots u_kv \in \Sigma^*$
in which case $\bigcup_{i=1}^k (u_{i+1}\cdots u_kv)^{-1}L_n =  (u_{j+1}\cdots u_kv)^{-1}L_n$.
Hence $\bigcup_{i=1}^k (u_{i+1}\cdots u_kv)^{-1}L_n$ is one of the $n+1$ quotients of $L_n$.

Each quotient of the product may therefore be written as $K \cup K'L_n$ for quotients $K'$ of $L'_m$ and $K$ of $L_n$.
Since $L'_m$ and $L_n$ have $m+1$ and $n+1$ quotients respectively, their product can have no more than $(m+1)(n+1)$ distinct quotients.
Recall that $L_n$ and $\emptyset$ are both quotients of $L_n$; in the case that $K'$ is final, the sets $L_n \cup K'L_n$ and $\emptyset \cup K'L_n$ are equal.
As $L'_m$ has at least one final quotient, we have the upper bound of $mn+m+n$. By comparison, if $L'_m$ and $L_n$ use the same alphabet then their product has complexity at most $m+n-1$~\cite{BDL15}.

The dialects  $(L'_m(a,b,-,d,e)\mid m \geq 4)$ and $(L_n(a,d,c,-,e) \mid n \ge 4)$ of the left ideal stream of Definition~\ref{def:leftideal}  meet the upper bound for product.
To prove this we apply the usual NFA construction for product.
This NFA is illustrated in Figure~\ref{fig:leftidealprod} for $m=n=4$.

\begin{figure}[h]
\unitlength 10pt
\begin{center}\begin{picture}(31,10.5)(-2,1.8)
\gasset{Nh=2,Nw=2,Nmr=1.2,ELdist=0.3,loopdiam=1.2}
{\small
\node(0')(2,8){$0'$}\imark(0')
\node(1')(6,8){$1'$}
\node(2')(10,8){$2'$}
\node(3')(14,8){$3'$}

\drawedge(0',1'){$e$}
\drawedge[ELdist=0.2](1',2'){$a,b$}
\drawedge(2',3'){$a$}
\drawloop(0'){$a,b,d$}
\drawloop(1'){$d,e$}
\drawloop(2'){$d$}
\drawloop(3'){$b$}
\drawedge[curvedepth=-2.1](2',1'){$b,e$}
\drawedge[curvedepth=2.5, ELside=r](3',1'){$a,e$}
\drawedge[curvedepth=3.8](3',0'){$d$}

\node(0)(14,4){$0$}
\node(1)(18,4){$1$}
\node(2)(22,4){$2$}
\node(3)(26,4){$3$}\rmark(3)

\drawedge(0,1){$e$}
\drawedge[ELdist=0.2](1,2){$a,d$}
\drawedge(2,3){$a$}
\drawloop[loopangle=225](0){$a,c,d$}
\drawloop(1){$c,e$}
\drawloop(2){$c$}
\drawloop(3){$d$}
\drawedge[curvedepth=-2.1](2,1){$d,e$}
\drawedge[curvedepth=2.5, ELside=r](3,1){$a,c,e$}

\drawedge(3',0){$\eps$}
}
\end{picture}\end{center}
\caption{NFA for product of left ideals.}
\label{fig:leftidealprod}
\end{figure}
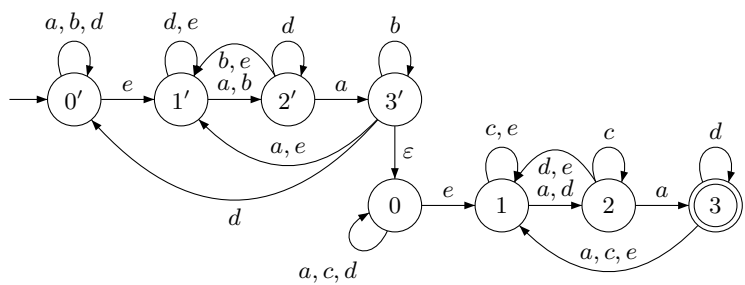

The subset construction  yields sets $\{p'\} \cup S$ where $p' \in Q'_m$ and $S \subseteq Q_n$, as well as sets $S \subseteq Q_n$.
However it is impossible to reach any set containing two or more elements of $Q_n \setminus \{0\}$ since they are only reachable from $0$ by applying $e$, and doing so sends all of $Q_n$ to $1$.
Moreover, in the DFA obtained from the NFA by the subset construction, the states $\{p', q\}$ are equivalent to $\{p', 0, q\}$ for each $p' \in Q'_m$ and $q \in Q_n$, and states $\{q\}$ and $\{0,q\}$ are similarly equivalent.
Hence we consider only the states $\{p'\}$ for $p' \in Q'_{m-1}$, $\{0,q\}$ for $q \in Q_n$, $\{p',0, q\}$ for $p' \in Q'_m$ and $q \in Q_n$, and $\emptyset$; a total of $mn+m+n$ states.

We check reachability of these states.
State $\{0'\}$ is initial and $\{p'\}$ is reached by $ea^{p-1}$ for $1 \le p \le m-2$.
State $\{(m-1)', 0\}$ is reached by $ea^{m-2}$, and $\{p',0\}$ is then reached by $a^p$ for $1 \le p \le m-1$, or by $d$ for $p'=0'$.
Applying $e$ from $\{(m-1)',0\}$ reaches $\{1',1\}$, and we reach $\{1', 0, 1\}$ by $ea^{(m-1)(n-1)}$.
State $\{p',0,2\}$ is then reached by $(ad)^{p-2}a$ for $2 \le p \le m-1$, and $\{1',0,2\}$ is reached from $\{1',0,1\}$ by $d$.
Now states of the form $\{p',0,q\}$ for $p \in Q'_m\setminus \{0'\}$ and $q \in Q_n \setminus \{0\}$ are reached by words in $a^*$ from states of the form $\{p',0, 2\}$.
State $\{0',0,q\}$ is reachable from $\{(m-1)',0,1\}$ by $da^{q-1}$, and
$\{0,q\}$ is reached from $\{0',0,1\}$ by $ca^{q-1}$.
Finally $\emptyset$ is reached from $\{0'\}$ by $c$. Thus, all states are reachable.

All non-empty states are distinguished from $\emptyset$ by $ea^{m-2}ea^{n-2}$.
States $\{p'_1\}$ and $\{p'_2\}$ where $p_1 < p_2$ are distinguished by $a^{m-1-p_2}ea^{n-2}$,
and states $\{0,q_1\}$ and $\{0,q_2\}$ where $q_1 < q_2$ are distinguished by $a^{n-1-q_2}$.
Two states which differ on $Q'_m$ are reduced to $\{p_1'\}$ and $\{p_2'\}$ by $b$,
and two states which differ on $Q_n$ are reduced to $\{0,q_1\}$ and $\{0,q_2\}$ by $c$ or $ac$.
Hence $(L'_m(a,b,-,d,e)$ and $(L_n(a,d,c,-,e)$ meet the upper bound for product.
\end{proof}

\subsection{Boolean Operations on Left Ideals}

\begin{theorem}[(Boolean Operations on Left Ideals)]
\label{thm:leftidealboolean}
For $m,n \ge 4$, the unrestricted complexities of boolean operations on left ideals are the same as those for arbitrary regular languages.
In particular, 
 the left ideals $L'_m(a,-,c,d,e)$ and $L_n(a,b,e,-,c)$  of Definition~\ref{def:leftideal}
meet the bound $(m+1)(n+1)$ for union and symmetric difference,
$L'_m(a,-,c,d,e)$ and
$L_n(a,-,e,-c)$ meet the bound $mn+m$ for difference, and 
$L'_m(a,-,c,-,e)$ and $L_n(a,-,e,-,c)$
meet the bound $mn$ for  intersection.


\end{theorem}
\begin{proof}
The upper bounds on the complexity of boolean operations for left ideals are the same as for regular languages.
We show that the left ideals $L'_m(a,-,c,d,e)$ and $L_n(a,b,e,-,c)$ of Definition~\ref{def:leftideal} meet these bounds.

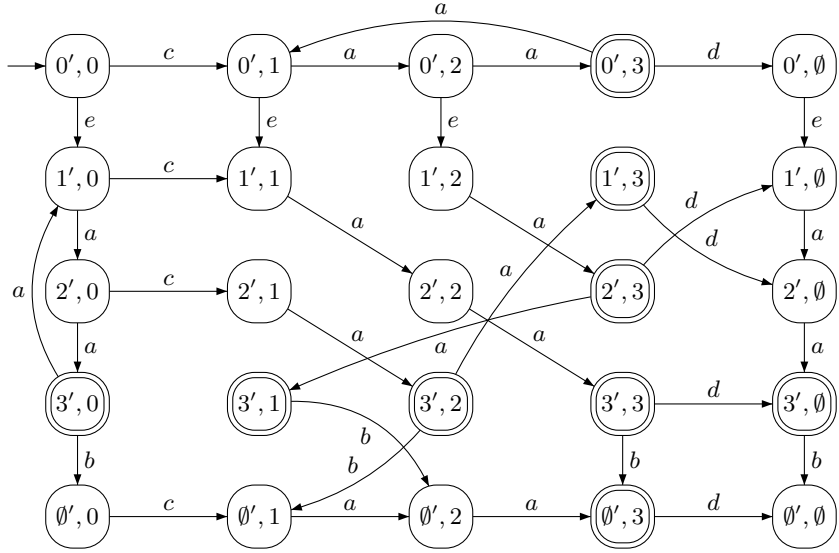
\begin{figure}[ht]
\unitlength 8.5pt
\begin{center}\begin{picture}(35,25)(0,-7)
\gasset{Nh=2.8,Nw=2.8,Nmr=1.2,ELdist=0.3,loopdiam=1.2}
	{\small
\node(0'0)(2,15){$0',0$}\imark(0'0)
\node(1'0)(2,10){$1',0$}
\node(2'0)(2,5){$2',0$}
\node(3'0)(2,0){$3',0$}\rmark(3'0)
\node(4'0)(2,-5){$\emp',0$}

\node(0'1)(10,15){$0',1$}
\node(1'1)(10,10){$1',1$}
\node(2'1)(10,5){$2',1$}
\node(3'1)(10,0){$3',1$}\rmark(3'1)
\node(4'1)(10,-5){$\emp',1$}

\node(0'2)(18,15){$0',2$}
\node(1'2)(18,10){$1',2$}
\node(2'2)(18,5){$2',2$}
\node(3'2)(18,0){$3',2$}\rmark(3'2)
\node(4'2)(18,-5){$\emp',2$}

\node(0'3)(26,15){$0',3$}\rmark(0'3)
\node(1'3)(26,10){$1',3$}\rmark(1'3)
\node(2'3)(26,5){$2',3$}\rmark(2'3)
\node(3'3)(26,0){$3',3$}\rmark(3'3)
\node(4'3)(26,-5){$\emp',3$}\rmark(4'3)

\node(0'4)(34,15){$0',\emp$}
\node(1'4)(34,10){$1',\emp$}
\node(2'4)(34,5){$2',\emp$}
\node(3'4)(34,0){$3',\emp$}\rmark(3'4)
\node(4'4)(34,-5){$\emp',\emp$}

\drawedge(0'0,0'1){$c$}
\drawedge(0'0,1'0){$e$}

\drawedge(1'0,2'0){$a$}
\drawedge(2'0,3'0){$a$}
\drawedge[curvedepth=2](3'0,1'0){$a$}

\drawedge(1'0,1'1){$c$}
\drawedge(2'0,2'1){$c$}

\drawedge(0'1,0'2){$a$}
\drawedge(0'2,0'3){$a$}
\drawedge[curvedepth=-2, ELside=r](0'3,0'1){$a$}

\drawedge(0'1,1'1){$e$}
\drawedge(0'2,1'2){$e$}

\drawedge(1'1,2'2){$a$}
\drawedge(2'2,3'3){$a$}

\drawedge(1'2,2'3){$a$}
\drawedge[curvedepth=-0.6](2'3,3'1){$a$}

\drawedge(2'1,3'2){$a$}
\drawedge[curvedepth=0.8](3'2,1'3){$a$}

\drawedge(0'3,0'4){$d$}
\drawedge[curvedepth=-1,ELdist=.4](1'3,2'4){$d$}
\drawedge[curvedepth=1, ELdist=.4](2'3,1'4){$d$}
\drawedge(3'3,3'4){$d$}
\drawedge(4'3,4'4){$d$}

\drawedge(3'0,4'0){$b$}
\drawedge[curvedepth=2, ELside=r](3'1,4'2){$b$}
\drawedge[curvedepth=1,ELdist=-1.2](3'2,4'1){$b$}
\drawedge(3'3,4'3){$b$}
\drawedge(3'4,4'4){$b$}

\drawedge(4'0,4'1){$c$}
\drawedge(4'1,4'2){$a$}
\drawedge(4'2,4'3){$a$}

\drawedge(0'4,1'4){$e$}
\drawedge(1'4,2'4){$a$}
\drawedge(2'4,3'4){$a$}
}
\end{picture}\end{center}
\caption{Direct product for union of $L'_4(a,b,c,-,e)$ and $L_4(a,d,e,-,c)$ shown partially.}
\label{fig:leftidealcross}
\end{figure}

The direct product for union is illustrated in Figure~\ref{fig:leftidealcross} for the case $m=n=4$.
Let $S = Q'_m \times Q_n$; when the languages have the same alphabet the direct product only contains the $mn$ states of $S$.
It was proved in~\cite{BDL15} that $L'_m(a,-,c,-,e) \circ L_n(a,-,e,-,c)$ has complexity $mn$ for each proper boolean operation $\circ \in \{\cup, \oplus, \setminus, \cap\}$;
hence all the states of $S$ are reachable and pairwise distinguishable by words in $\{a,c,e\}^*$ for each operation.
The remaining states $(p',\emptyset)$ and $(\emptyset',q)$ in the unrestricted direct product are reachable using $b$ and $d$ -- which are not present in both alphabets.


It remains to determine which states are distinguishable for each operation.
Let $H = \{(\emptyset',q)\mid q \in Q_n\}$ and $V = \{(p', \emptyset) \mid p' \in Q'_m\}$.

\bd

\item[\bf Union] The final states are $\{((m-1)',q) \mid q \in Q_n \cup \{\emptyset\}\}$ and $\{(p', n-1) \mid p' \in Q'_m \cup \{\emptyset'\}\}$.
Every state in $V$ accepts a word with a $b$, whereas  no state in $H$ accepts such words.
Similarly, every state in $H$ accepts a word with a $d$, whereas  no state in $V$ accepts such words.
Every state in $Q'_m \times Q_n$ accepts a word with a $b$ and a word with a $d$. State $(\emp',\emp)$ accepts no words at all.
Hence any two states chosen from different sets (the sets being $S$, $H$, $V$,  and $\{(\emp',\emp)\}$) are distinguishable.
States in $H$ are distinguishable by words in $a^*$, as are those in $V$.
Therefore all $(m+1)(n+1)$ states are pairwise distinguishable.

\item[\bf Symmetric Difference]
The final states here are all the final states for union except $( (m-1)',n-1 )$. The rest of the argument is the same as for union.

\item[\bf Difference]
The final states now are $\{((m-1)', q) \mid q\neq n-1\}$.
The $n$ states of the form $(\emp',q)$, $q \in Q_n$, are now equivalent to the empty state $\{(\emp',\emp)\}$. 
The remaining states are pairwise distinguishable by the arguments used for union. Hence we have $mn+m+1$ distinguishable states. However, 
since $b$ is not in the alphabet of the difference it can be omitted along with the empty states  $(\emp',q)$, $q\in Q_n \cup \{\emp\}$ (reachable only by $b$), giving  a bound of $mn+n$. 
This bound can be reached by $L'_m(a,-,c,d,e)$ and
$L_n(a,-,e,-c)$.

\item[\bf Intersection]
Here only $((m-1)', n-1 )$ is final and all states $(p', \emp)$, $p' \in Q'_m$, and $(\emp',q)$, $q\in Q_n$ are equivalent to $\{(\emp',\emp)\}$, leaving $mn+1$ distinguishable states.
However, the empty state can be reached only by $b$ or $d$, which are not in the alphabet of $L'_m(a,-,c,d,e)\cap L_n(a,b,e,-,c)$. Hence the correct bound is $mn$ and it is reached by $L'_m(a,-,c,-,e)$ and
$L_n(a,-,e,-c)$.
\ed
The complexities of all ten boolean functions on left ideals are given in Remark~\ref{rem:boolean}.
\end{proof}

\subsection{Most Complex Left Ideals}

We now update the results from~\cite{BDL15} to include the unrestricted cases.

\begin{theorem}[(Most Complex Left Ideals)]
\label{thm:leftidealmain}
For each $n\ge 4$, the DFA of Definition~\ref{def:leftideal} is minimal and its 
language $L_n(a,b,c,d,e)$ has complexity $n$.
The stream $(L_m(a,b,c,d,e) \mid m \ge 4)$ with some dialect streams
is most complex in the class of left ideals.
In particular, it meets all the complexity bounds below, which are maximal for left ideals.
In several cases the bounds can be met with a reduced alphabet.
\begin{enumerate}
\item
The syntactic semigroup of $L_n(a,b,c,d,e)$ has cardinality $n^{n-1}+n-1$. Moreover, fewer than five inputs do not suffice to meet this bound.  
\item
Each quotient of $L_n(a,-,-,d,e)$ has complexity $n$.
\item
The reverse of $L_n(a,-,c,d,e)$ has complexity $2^{n-1}+1$, and $L_n(a,-,c,d,e)$ has $2^{n-1}+1$ atoms.
\item
For each atom $A_S$ of $L_n(a,b,c)$, the complexity $\kappa(A_S)$ satisfies:
\begin{equation*}
	\kappa(A_S) =
	\begin{cases}
		n, 			& \text{if $S =Q_n$;}\\
		2^{n-1}, 		& \text{if $S =\emp$;}\\
		1+ \sum_{x=1}^{|S|}\sum_{y=1}^{n-|S|} \binom{n-1}{x}\binom{n-1-x}{y},
		 			& \text{otherwise.}
		\end{cases}
\end{equation*}
\item
The star of $L_n(a,-,-,-,e)$ has complexity $n+1$.
\item Product
	\be 
	\item
	Restricted case:\\
	The product $L_m(a,-,-,-,e) L_n(a,-,-,-,e)$ has complexity $m+n-1$.
	\item
	Unrestricted case:\\
	The product $L'_m(a,b,-,d,e) L_n(a,d,c,-,e)$ has complexity $mn+m+n$.
	\ee
\item Boolean operations
	\be
	\item Restricted case:\\
	For any proper binary boolean function $\circ$, the complexity of  $L_m(a,-,c,-,e)\\ 		\circ L_n(a,-,e,-,c)$ is $mn$.
	\item Unrestricted case:\\
	The complexity of $L'_m(a,-,c,d,e) \circ L_n(a,b,e,-,c)$  is the same as  for arbitrary regular languages:
$(m+1)(n+1)$ if $\circ\in \{\cup,\oplus\}$, 
$mn+m$ if $\circ=\setminus$, and $mn$ if $\circ=\cap$.
The bound for difference is also met by $L'_m(a,-,c,d,e) \setminus L_n(a,-,e,-,c)$  and the bound for intersection  by 
$L'_m(a,-,c,-,e) \cap L_n(a,-,e,-,c)$.
\ee
\end{enumerate}
\end{theorem}
\begin{proof}
The proofs for the restricted cases can be found in~\cite{BDL15}, and
the claims about the unrestricted complexities are proved in Theorems~\ref{thm:leftidealproduct} and~\ref{thm:leftidealboolean}.
\end{proof}

\section{Two-Sided Ideals}
The following stream of two-sided ideals was defined in~\cite{BrYe11}, where it was conjectured that the DFAs in this stream have maximal transition semigroups. This conjecture was proved in~\cite{BrSz14a},  and the stream was shown to be most complex for restricted operations in~\cite{BDL15}. We prove that it is also most complex in the unrestricted case.

\begin{definition}
\label{def:2sided}
For $n\ge 5$, let  
$\cD_n =\cD_n(a,b,c,d,e,f)= (Q_n,\Sig,\delta, 0,\{n-1\})$, where
$\Sig=\{a,b,c,d,e,f\}$, 
and $\delta_n$ is defined by the transformations
$a \colon (1,2,\ldots,n-2)$,
$b \colon (1,2)$,
$c \colon (n-2\to 1)$,
$d \colon (n-2\to 0)$,
$e \colon (Q_{n-1}\to 1)$,
and $f \colon (1 \to n-1)$.
The structure of  $\cD_n(a,b,c,d,e,f)$ is shown in Figure~\ref{fig:2sided}. 
\end{definition}

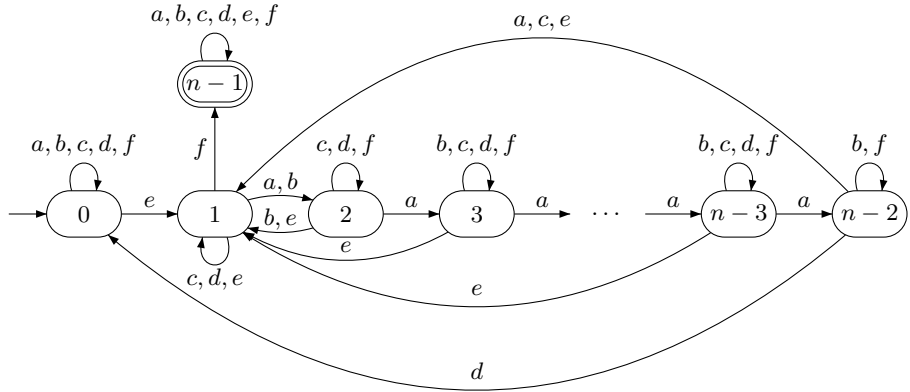
\begin{figure}[h]
\unitlength 7pt
\begin{center}\begin{picture}(43,19)(0,-1)
\gasset{Nh=2.5,Nw=4,Nmr=1.25,ELdist=0.4,loopdiam=1.5}
{\small
\node(n-1)(8,14){$n-1$}\rmark(n-1)
\drawloop(n-1){$a,b,c,d,e,f$}

\node(0)(1,7){$0$}\imark(0)
\node(1)(8,7){$1$}
\node(2)(15,7){$2$}
\node(3)(22,7){$3$}
\node[Nframe=n](4dots)(29,7){$\dots$}
\node(n-3)(36,7){$n-3$}
\node(n-2)(43,7){$n-2$}
\drawedge(1,n-1){$f$}
\drawedge(0,1){$e$}
\drawloop(0){$a,b,c,d,f$}
\drawloop[loopangle=270,ELdist=.2](1){$c,d,e$}
\drawedge[curvedepth= 1,ELdist=.1](1,2){$a,b$}
\drawedge[curvedepth= 1,ELdist=-1.3](2,1){$b,e$}
\drawloop(2){$c,d,f$}
\drawedge(2,3){$a$}
\drawedge[curvedepth= 2.5,ELdist=-1](3,1){$e$}
\drawedge(3,4dots){$a$}
\drawedge(4dots,n-3){$a$}
\drawloop(3){$b,c,d,f$}
\drawloop(n-3){$b,c,d,f$}
\drawedge(n-3,n-2){$a$}
\drawedge[curvedepth= 5,ELdist=-1.2](n-3,1){$e$}
\drawedge[curvedepth= -9.5,ELdist=-1.2](n-2,1){$a,c,e$}
\drawedge[curvedepth= 9.5,ELdist=-1.5](n-2,0){$d$}
\drawloop(n-2){$b,f$}
}
\end{picture}\end{center}
\caption{Minimal DFA $\cD_n(a,b,c,d,e,f)$  of Definition~\ref{def:2sided}.}
\label{fig:2sided}
\end{figure}

\subsection{Product of Two-Sided Ideals}

\begin{theorem}[(Product of Two-Sided Ideals)]
\label{thm:2sidedidealproduct}
For $m,n \ge 5$, let $L'_m$ (respectively, $L_n$) be any two-sided ideal with $m$ (respectively, $n$) quotients over an alphabet $\Sig'$, (respectively, $\Sig$). 
Then $\kappa(L'_m L_n) \le m+2n$, and this bound is met by the two-sided ideals $L'_m(a,b,-,-,e,f)$ and  $L_n(a,c,-,-,e,f)$ of Definition~\ref{def:2sided}.
\end{theorem}

\begin{proof}
We first derive an upper bound on the complexity of $L'_mL_n$, where $L'_m$ is a right ideal over $\Sigma'$ and $L_n$ is a left ideal over $\Sigma$;
as two-sided ideals are both right and left ideals, this upper bound will apply.
In taking their product, we consider the languages over the union of their alphabets.
 If $L'_m$ and $L_n$ use different alphabets then they each have an empty quotient when taken over the combined alphabet;
 hence $L'_m$ and $L_n$ have up to $m+1$ and $n+1$ quotients respectively.

We count the number of quotients $w^{-1}(L'_mL_n)$ where $w \in (\Sigma' \cup \Sigma)^*$.
Write $w =  u_1u_2 \cdots u_kv$ where $u_1\cdots u_i \in L'_m$ for $i=1,\dots, k$ and no other prefix of $w$ is in $L'_m$.
Then $wx \in L'_mL_n$ if and only if $u_{i+1} \cdots u_k v x \in L_n$ for some $1 \le i \le k$, or $x=yz$ is such that $wy \in L'_m$ and $z \in L_n$; 
hence 
$$w^{-1}(L'_mL_n) = \bigcup_{i=1}^k (u_{i+1}\cdots u_kv)^{-1}L_n \cup (w^{-1}L'_m)L_n.$$

Consider the case where $w \in {\Sigma'}^*$.
If $k=0$, then $w^{-1}(L'_mL_n)$ simplifies to $(v^{-1}L'_m)L_n$; there are at most $m-1$ quotients of this form since $v^{-1}L'_m$ is one of the non-final quotients of $L'_m$.
Assume $k \ge 1$; then $w^{-1}L'_m = {\Sigma'}^*$ since $L'_m$ is a right ideal with respect to $\Sigma'$. Hence
$$w^{-1}(L'_mL_n) = \bigcup_{i=1}^k (u_{i+1}\cdots u_kv)^{-1}L_n \cup {\Sigma'}^*L_n.$$
Moreover there are only $n+1$ possible values for $\bigcup_{i=1}^k (u_{i+1}\cdots u_kv)^{-1}L_n$:
Notice that $v^{-1}L_n \subseteq (u_kv)^{-1}L_n \subseteq  (u_{k-1}u_kv)^{-1}L_n \subseteq \dots \subseteq (u_{i+1}\cdots u_kv)^{-1}L_n$ as long as $u_{i+1}\cdots u_k \in \Sigma^*$, since $L_n$ is a left ideal with respect to $\Sigma$.
Alternatively if $u_{i+1}\cdots u_kv$ contains a letter from $\Sigma' \setminus \Sigma$ then $(u_{i+1}\cdots u_kv)^{-1}L_n = \emptyset$.
Thus, either $\bigcup_{i=1}^k (u_{i+1}\cdots u_kv)^{-1}L = \emptyset$, or there exists a minimal $j \in \{1, \dots, k\}$ such that $u_{j+1}\cdots u_kv \in \Sigma^*$
in which case $\bigcup_{i=1}^k (u_{i+1}\cdots u_kv)^{-1}L_n =  (u_{j+1}\cdots u_kv)^{-1}L_n$.
Hence $\bigcup_{i=1}^k (u_{i+1}\cdots u_kv)^{-1}L_n$ is one of the $n+1$ quotients of $L_n$.

Now $w^{-1}(L'_mL_n) = K \cup {\Sigma'}^*L_n$, which has at most $n+1$ distinct values as $K$ is a quotient of $L_n$.
However $L_n \cup {\Sigma'}^*L_n = \emptyset \cup {\Sigma'}^*L_n$, so there are only $n$ distinct values.
So far, we have counted a total of $(m-1)+n$ quotients of the product.

Now suppose $w \not\in {\Sigma'}^*$; then $w^{-1}L'_m = \emptyset$. By the same reasoning we have
$w^{-1}(L'_mL_n) = K \cup (w^{-1}L'_m)L_n = K$ for some quotient $K$ of $L_n$.
There are at most $n+1$ quotients of $L_n$, yielding an upper bound of $(m-1)+n + (n+1) = m+2n$ quotients of $L'_mL_n$.
\medskip

We prove this bound is tight using the dialect streams $(L'_m(a,b,-,-,e,f) \mid m \ge 5)$ and $(L_n(a,c,-,-,e,f)  \mid n \ge 5)$ of Definition~\ref{def:2sided}.
We apply the usual NFA construction for product.
This NFA is illustrated in Figure~\ref{fig:2sidedidealprod} for $m=n=5$.

\begin{figure}[h]
\unitlength 10pt
\begin{center}\begin{picture}(31,17)(-2,-4)
\gasset{Nh=2.2,Nw=2.2,Nmr=1.25,ELdist=0.4,loopdiam=1.5}
{\small
\node(0')(1,7){$0'$}\imark(0')
\node(1')(5,7){$1'$}
\node(2')(9,7){$2'$}
\node(3')(13,7){$3'$}
\node(4')(17,7){$4'$}
\node(5')(5,3){$5'$}

\drawloop[loopangle=180,ELdist=.2](5'){$a,b,e,f$}
\drawedge(1',5'){$f$}
\drawedge(0',1'){$e$}
\drawloop(0'){$a,b,f$}
\drawloop(1'){$e$}
\drawedge[curvedepth= 1,ELdist=.1](1',2'){$a,b$}
\drawedge[curvedepth= 1,ELdist=-1.2](2',1'){$b,e$}
\drawloop(2'){$f$}
\drawedge(2',3'){$a$}
\drawedge[curvedepth=2.2,ELside=r](3',1'){$e$}
\drawedge(3',4'){$a$}
\drawloop(3'){$b,f$}
\drawedge[curvedepth=-5,ELside=r](4',1'){$a,e$}
\drawloop(4'){$b,f$}

\node(0)(9,-1){$0$}
\node(1)(13,-1){$1$}
\node(2)(17,-1){$2$}
\node(3)(21,-1){$3$}
\node(4)(13,3){$4$}\rmark(4)

\drawloop[loopangle=20,ELdist=.2](4){$a,c,e,f$}
\drawedge[ELside=r](1,4){$f$}
\drawedge(0,1){$e$}
\drawloop[loopangle=270](0){$a,c,f$}
\drawloop[loopangle=270](1){$e$}
\drawedge[curvedepth= 1,ELdist=.1](1,2){$a,c$}
\drawedge[curvedepth= 1,ELside=r](2,1){$c,e$}
\drawloop[loopangle=270](2){$f$}
\drawedge(2,3){$a$}
\drawedge[curvedepth=-3](3,1){$a,e$}
\drawloop[loopangle=270](3){$c,f$}

\drawedge[curvedepth=-1.2](5',0){$\eps$}
}
\end{picture}\end{center}
\caption{NFA for product of two-sided ideals.}
\label{fig:2sidedidealprod}
\end{figure}
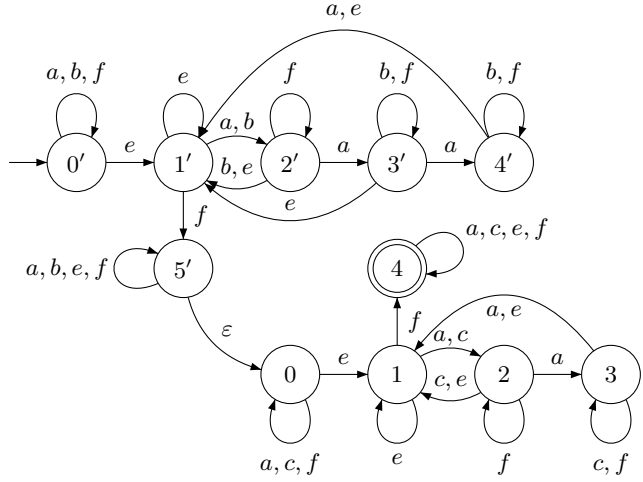

The subset construction yields sets $\{p'\} \cup S$ where $p' \in Q'_m$ and $S \subseteq Q_n$ as well as sets $S \subseteq Q_n$, although many sets of this form are not reachable here.
In $\cD'_m$, states $p' \in Q'_{m-1}$ are not reachable from $(m-1)'$ and so all sets $\{p'\} \cup S$ where $S \not= \emptyset$ are unreachable.
In $\cD_n$, states of $\{1, \dots, n-2\}$ are only reachable from $0$ using $e$, which maps $Q_n \setminus \{n-1\}$ to $1$, and no other state is reachable from $n-1$;
this restricts $S$ to contain only one state from $\{1, \dots, n-2\}$.
Thus, the only potentially reachable sets are
$\{p'\}$ for $p' \in Q'_{m-1}$,
$\{(m-1)',0, q\}$ and $\{(m-1)',0, q, n-1\}$ for $q \in Q_{n-1}$, 
$\{0, q\}$ and $\{0, q, n-1\}$ for $q \in Q_{n-1}$, 
$\{q\}$ and $\{q, n-1\}$ for $q \in Q_{n-1}$, and $\emptyset$.

We reach sets $\{p'\}$ by $ea^{p-1}$ for $1 \le p \le m-2$, and $\{(m-1)', 0\}$ by $ef$.
From $\{(m-1)', 0\}$ we reach $\{(m-1)', 0, q\}$ by $ea^{q-1}$ for $1 \le q \le n-2$.
Set $\{(m-1)',0,n-1\}$ is reached from $\{(m-1)',0,1\}$ by $f$, and $\{(m-1)',0,q,n-1\}$ is then reached by $ea^{q-1}$.
From $\{(m-1)',0\} \cup S$ we reach $\{0\} \cup S$ by $c^2$.
Set $\{q\}$ is reached from $\{0\}$ by $ea^{q-1}$ for $q \in \{1, \dots, n-2\}$,
and $\{q, n-1\}$ is reached from $\{0, n-1\}$ by the same word.
Finally, $\emptyset$ is reached from $\{0'\}$ by $c$.

In the DFA obtained from the NFA by the subset construction, many of these subsets represent equivalent states. All sets containing $n-1$ and not $(m-1)'$ accept $\{a,c,e,f\}^*$ and are mapped to $\emptyset$ by $b$; hence they are equivalent.
Similarly all sets containing $(m-1)'$ and $n-1$ accept $\{a,c,e,f\}^*$ and are mapped to $\{(m-1)',0\}$ by $b$; hence they are equivalent.
Moreover sets $\{0,q\}$ and $\{q\}$ are equivalent for $1 \le q \le n-2$, since any word that maps $0$ to $n-1$ also maps $q$ to $n-1$.
Therefore we only need to consider sets
$\{p'\}$ for $p' \in Q'_{m-1}$,
$\{(m-1)', 0, q\}$ for $q \in Q_{n-1}$, and
$\{0, q\}$ for $q \in Q_{n-1}$, as well as $\{(m-1)', n-1\}$, $\{n-1\}$, and $\emptyset$.
Every reachable set is equivalent to one of these $m+2n$ sets. We prove that they are pairwise distinguishable.

All non-empty sets are distinguished from $\emptyset$ by $efef$.
States $\{p'_1\}$ and $\{p'_2\}$ where $p_1 < p_2$ are distinguished by $a^{m-1-p_2}fef$.
Sets $\{(m-1)', 0, q_1\}$ and $\{(m-1)', 0, q_2\}$ for $q_1 < q_2$ are distinguished by $a^{n-1-q_2}f$;
sets $\{0, q_1\}$ and $\{0, q_2\}$ are similarly distinguished.
State $\{p\}$ is distinguished from any state containing $q \in Q_n$ by $ef$.
Finally $\{(m-1)', 0\} \cup S_1$ is distinguished from $S_2 \subseteq Q_n$ by $bef$.
Thus all $m+2n$ states are distinguishable.
\end{proof}

\subsection{Boolean Operations on Two-Sided Ideals}

\begin{theorem}[(Boolean Operations on Two-Sided Ideals)]
\label{thm:2sidedidealboolean}
For $m,n \ge 5$, the unrestricted complexities of boolean operations on two-sided ideals are the same as those for arbitrary regular languages.
In particular, 
 the two-sided ideals $L'_m(a,b,c,-,e,f)$ and $L_n(a,e,d,-,b,f)$ of 
 Definition~\ref{def:2sided}
meet the bound $(m+1)(n+1)$ for union and symmetric difference,
$L'_m(a,b,c,-,e,f)$ and
$L_n(a,e,-,-,b,f)$ meet the bound $mn+m$ for difference, and 
$L'_m(a,b,-,-,e,f)$ and $L_n(a,e,-,-,b,f)$
meet the bound $mn$ for  intersection.


\end{theorem}
\begin{proof}
The upper bounds on the complexity of boolean operations for two-sided ideals are the same as for regular languages.
We show that the two-sided ideals $L'_m(a,b,c,-,e,f)$ and $L_n(a,e,d,-,b,f)$ of Definition~\ref{def:2sided} meet these bounds.

\begin{figure}[ht]
\unitlength 8.5pt
\begin{center}\begin{picture}(34,33)(0,-3.5)
\gasset{Nh=2.8,Nw=2.8,Nmr=1.2,ELdist=0.3,loopdiam=1.2}
	{\small
\node(0'0)(2,25){$0',0$}\imark(0'0)
\node(1'0)(2,20){$1',0$}
\node(2'0)(2,15){$2',0$}
\node(3'0)(2,10){$3',0$}
\node(4'0)(2,5){$4',0$}\rmark(4'0)
\node(5'0)(2,0){$\emp',0$}

\node(0'1)(8,25){$0',1$}
\node(1'1)(8,20){$1',1$}
\node(2'1)(8,15){$2',1$}
\node(3'1)(8,10){$3',1$}
\node(4'1)(8,5){$4',1$}\rmark(4'1)
\node(5'1)(8,0){$\emp',1$}

\node(0'2)(14,25){$0',2$}
\node(1'2)(14,20){$1',2$}
\node(2'2)(14,15){$2',2$}\
\node(3'2)(14,10){$3',2$}
\node(4'2)(14,5){$4',2$}\rmark(4'2)
\node(5'2)(14,0){$\emp',2$}

\node(0'3)(20,25){$0',3$}
\node(1'3)(20,20){$1',3$}
\node(2'3)(20,15){$2',3$}
\node(3'3)(20,10){$3',3$}
\node(4'3)(20,5){$4',3$}\rmark(4'3)
\node(5'3)(20,0){$\emp',2$}

\node(0'4)(26,25){$0',4$}\rmark(0'4)
\node(1'4)(26,20){$1',4$}\rmark(1'4)
\node(2'4)(26,15){$2',4$}\rmark(2'4)
\node(3'4)(26,10){$3',4$}\rmark(3'4)
\node(4'4)(26,5){$4',4$}\rmark(4'4)
\node(5'4)(26,0){$\emp',4$}\rmark(5'4)

\node(0'5)(32,25){$0',\emp$}
\node(1'5)(32,20){$1',\emp$}
\node(2'5)(32,15){$2',\emp$}
\node(3'5)(32,10){$3',\emp$}
\node(4'5)(32,5){$4',\emp$}\rmark(4'5)
\node(5'5)(32,0){$\emp',\emp$}

\drawedge(0'0,0'1){$b$}
\drawedge(0'0,1'0){$e$}

\drawedge(1'0,2'0){$a$}
\drawedge(2'0,3'0){$a$}

\drawedge[ELpos=35](1'0,2'1){$b$}
\drawedge[ELpos=35, ELside=r](2'0,1'1){$b$}
\drawedge(3'0,3'1){$b$}

\drawedge(0'1,0'2){$a$}
\drawedge(0'2,0'3){$a$}

\drawedge[ELpos=35, ELside=r](0'1,1'2){$e$}
\drawedge[ELpos=35](0'2,1'1){$e$}
\drawedge(0'3,1'3){$e$}

\drawedge(1'1,2'2){$a$}
\drawedge(2'2,3'3){$a$}
\drawedge(1'2,2'3){$a$}
\drawedge(2'1,3'2){$a$}

\drawedge[curvedepth=-2.5, ELside=r](1'0,4'0){$f$}
\drawedge[curvedepth=2.5](0'1,0'4){$f$}
\drawedge[curvedepth=2.5](4'1,4'4){$f$}
\drawedge[curvedepth=-2.5, ELside=r](1'4,4'4){$f$}
\drawedge[curvedepth=2.5](1'5,4'5){$f$}
\drawedge[curvedepth=-2.5, ELside=r](5'1,5'4){$f$}

\drawedge(4'0,4'1){$b$}
\drawedge(4'1,4'2){$a$}
\drawedge(4'2,4'3){$a$}
\drawedge(0'4,1'4){$e$}
\drawedge(1'4,2'4){$a$}
\drawedge(2'4,3'4){$a$}

\drawedge(0'4,0'5){$c$}
\drawedge(1'4,1'5){$c$}
\drawedge[ELpos = 35, ELside=r](2'4,1'5){$c$}
\drawedge(3'4,3'5){$c$}
\drawedge(4'4,4'5){$c$}
\drawedge(5'4,5'5){$c$}

\drawedge(4'0,5'0){$d$}
\drawedge(4'1,5'1){$d$}
\drawedge[ELpos = 35](4'2,5'1){$d$}
\drawedge(4'3,5'3){$d$}
\drawedge(4'4,5'4){$d$}
\drawedge(4'5,5'5){$d$}

\drawedge(5'0,5'1){$b$}
\drawedge(5'1,5'2){$a$}
\drawedge(5'2,5'3){$a$}

\drawedge(0'5,1'5){$e$}
\drawedge(1'5,2'5){$a$}
\drawedge(2'5,3'5){$a$}
}
\end{picture}\end{center}
\caption{Direct product for union of $L'_5(a,b,c,-,e,f)$ and $L_5(a,e,d,-,b,f)$ shown partially.}
\label{fig:2sidedidealcross}
\end{figure}
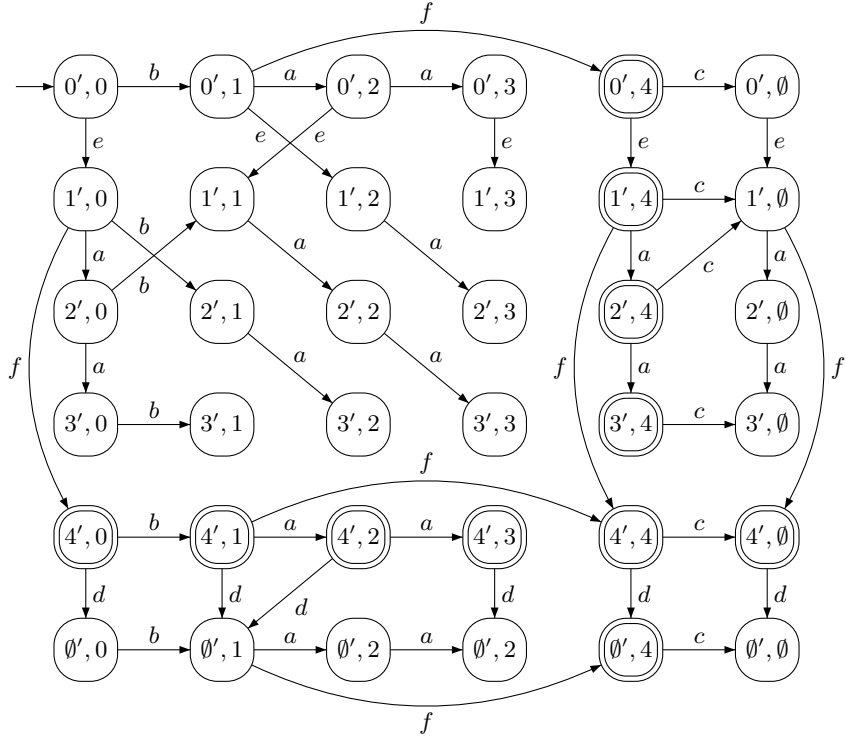

The direct product for union is illustrated in Figure~\ref{fig:2sidedidealcross}.
We check that all $(m+1)(n+1)$ states are reachable.
First consider the states of $Q'_{m-1} \times Q_{n-1}$.
State $(0',0)$ is initial, $(0', q)$ is reached by $b a^{q-1}$, and $(p',0)$ is reached by $e a^{p-1}$.
State $(1',q)$ is reached from $(0',q)$ by $e^2$, and $(p', 1)$ is reached from $(p', 0)$ by $b^2$.
State $(p', q)$ with $q \ge p$ is reached from $(1', q-(p-1))$ by $a^{p-1}$.
By symmetry all states $(p', q)$ with $p \ge q$ are also reachable; hence the states of $Q'_{m-1} \times Q_{n-1}$ are reachable.
The remaining states of $\{(m-1)', \emptyset'\} \times Q_n\cup\{\emptyset\}$ and $Q'_m \cup \{\emptyset'\} \times \{n-1, \emptyset\}$ are easily seen to be reachable in Figure~\ref{fig:2sidedidealcross}.

It remains to determine which states are distinguishable for each operation.
Let $H = \{(\emptyset',q)\mid q \in Q_n\}$ and $V = \{(p', \emptyset) \mid p' \in Q'_m\}$.

\bd

\item[\bf Union] The final states are $\{((m-1)',q) \mid q \in Q_n \cup \{\emptyset\}\}$ and $\{(p', n-1) \mid p' \in Q'_m \cup \{\emptyset'\}\}$.
Every state in $V$ accepts a word with a $c$, whereas no state in $H$ accepts such words.
Similarly, every state in $H$ accepts a word with a $d$, whereas  no state in $V$ accepts such words.
Every state in $Q'_m \times Q_n$ accepts a word with a $b$ and a word with a $d$. State $(\emp',\emp)$ accepts no words at all.
Hence any two states chosen from different sets (the sets being $Q'_m \times Q_n$, $H$, $V$,  and $\{(\emp',\emp)\}$) are distinguishable.
States in $H$ are distinguishable by words in $a^*f$, and those in $V$ by words in $a^*f$.
States that differ in the first coordinate are sent to distinct states of $V$ by $c$; the exceptions are the pair of states $(1', q_1)$ and $(2',q_2)$, which are sent to distinct states of $V$ by $ac$. Similarly, states that differ in the second coordinate are sent to distinct states of $H$ by $d$ or $ad$.
Therefore all $mn+m+n+1$ states are pairwise distinguishable.

\item[\bf Symmetric Difference]
The final states here are all the final states for union except $( (m-1)',n-1 )$. The rest of the argument is the same as for union.

\item[\bf Difference]
The final states now are $\{((m-1)', q) \mid q\neq n-1\}$.
The $n$ states of the form $(\emp',q)$, $q \in Q_n$ are now equivalent to the empty state $\{(\emp',\emp)\}$;
all other states are non-empty as they accept $cef$.
States that differ in the first coordinate are distinguished by words in $ca^*f$;
the exceptions are the pairs of states $(1', q_1)$ and $(2',q_2)$ which are distinguished by words in $aca^*f$.
States $((m-1)', q)$ for $q \in Q_n$ are distinguished by words in $a^*f$, and
$((m-1)', \emptyset)$ is distinguished from $((m-1)', q)$ by $bf$.
States that differ in the second coordinate are mapped by either $ef$ or $eaf$ to distinct states of $\{((m-1)', q) \mid q \in Q_n \cup \{\emptyset\}\}$, and are hence distinguishable.
Hence we have $mn+m+1$ distinguishable states.
However, $d$ is not in the alphabet of the difference, and so $d$ and the empty state can be removed, giving bound $mn+m$ reached by 
$L'_m(a,b,c,-,e,f)$ and $L_n(a,e,-,-,b,f)$.

\item[\bf Intersection]
Here only $((m-1)', n-1 )$ is final and all states $(p', \emp)$, $p' \in Q'_m$, and $(\emp',q)$, $q\in Q_n$ are equivalent to $\{(\emp',\emp)\}$, leaving $mn+1$ states.
All other states are non-empty as they accept $efbf$.
States $(p', n-1)$ for $p' \in Q'_m$, as well as 
states $((m-1)', q)$ for $q \in Q_n$ are distinguished by words in $a^*f$.
States that differ in the first coordinate are mapped by either $bf$ or $baf$ to distinct states of $\{((m-
1)', q) \mid q \in Q_n\}$, and are hence distinguishable.
States that differ in the second coordinate are mapped by either $ef$ or $eaf$ to distinct states of $\{(p', n-1) \mid p' \in Q'_m\}$, and are hence distinguishable.
However, $c$ and $d$ are not in the alphabet of the intersection, and so $c$, $d$ and the empty state can be removed, giving bound $mn$ reached by 
$L'_m(a,b,-,-,e,f)$ and $L_n(a,e,-,-,b,f)$.
\ed
The complexities of all ten boolean functions on two-sided ideals are given in Remark~\ref{rem:boolean}.
\end{proof}
\subsection{Most Complex Two-Sided Ideals}

\begin{theorem}[(Most Complex Two-Sided Ideals)]
\label{thm:twosidedidealmain}
For each $n\ge 5$, the DFA of Definition~\ref{def:leftideal} is minimal and its 
language $L_n(a,b,c,d,e,f)$ has complexity $n$.
The stream $(L_n(a,b,c,d,e,f) \mid n \ge 5)$  with some dialect streams
is most complex in the class of regular two-sided ideals.
In particular, this stream meets all the complexity bounds listed below, which are maximal for two-sided ideals.
In several cases the bounds can be met with a reduced alphabet.

\begin{enumerate}
\item
The syntactic semigroup of $L_n(a,b,c,d,e,f)$ has size $n^{n-2}+(n-2)2^{n-2}+1$.  Moreover, fewer than six inputs do not suffice to meet this bound.
\item
Each quotient of $L_n(a,-,-,d,e,f)$ has complexity $n$.
\item
The reverse of the language $L_n(a,-,-,d,e,f)$ has complexity $2^{n-1}+1$, and $L_n(a,-,-,d,e,f)$ has $2^{n-1}+1$ atoms.
\item
For each atom $A_S$ of $L_n(a,b,c,d,e,f)$, the complexity $\kappa(A_S)$ satisfies:
\begin{equation*}
	\kappa(A_S) =
	\begin{cases}
		 n, 			& \text{if $S=Q_n$;}\\
		 2^{n-2}+n-1,		& \text{if $S=Q_n\setminus \{1\} $;}\\
		  1 + \sum_{x=1}^{|S|}\sum_{y=1}^{n-|S|}\binom{n-2}{x-1}\binom{n-x-1}{y-1},
		 			& \text{otherwise.}
		\end{cases}
\end{equation*}
\item
The star of $L_n(a,-,-,-,e,f)$ has complexity $n+1$.
\item Product
	\be
	\item 
	Restricted case:\\
	The product $L_m(a,-,-,-,e,f) L_n(a,-,-,-,e,f)$ has complexity $m+n-1$.
	\item
	Unrestricted case:\\
	The product $L'_m(a,b,-,-,e,f) L_n(a,c,-,-,e,f)$ has complexity $m+2n$.
	\ee
\item Boolean operations
	\be
	\item
	Restricted case:\\
	For any proper binary boolean function $\circ$, the complexity of $L_m(a,b,-,d,e,f) \\ \circ L_n(b,a,-,d,e,f)$
is maximal. In particular,
	\be
	\item
	 $L_m(a,b,-,d,e,f) \cap L_n(b,a,-,d,e,f)$ has complexity $mn$, as does\\ 
	 $L_m(a,b,-,d,e,f) \oplus L_n(b,a,-,d,e,f)$.
	 \item
	$L_m(a,b,-,d,e,f) \setminus L_n(b,a,-,d,e,f)$ has complexity 	$mn
	- (m-1)$.
	\item
	$L_m(a,b,-,d,e,f) \cup L_n(b,a,-,d,e,f)$ has complexity 	$mn
	- (m+n-2)$.
	\ee

	\item
	Unrestricted case:\\
	The complexity of $L'_m(a,b,c,-,e,f) \circ L_n(a,e,d,-,b,f)$  is the same as  for arbitrary regular languages:
$(m+1)(n+1)$ if $\circ\in \{\cup,\oplus\}$, 
$mn+m$ if $\circ=\setminus$, and $mn$ if $\circ=\cap$.
The bound for difference is also met by $L'_m(a,b,c,-,e,f) \setminus L_n(a,e,-,-,b,f)$ and the bound for intersection  by $L'_m(a,b,-,-,e,f) \cap L_n(a,e,-,-,b,f)$.

	\ee
\end{enumerate}
\end{theorem}

\begin{proof}
The proofs for the restricted cases can be found in~\cite{BDL15}, and
the claims about the unrestricted complexities are proved in Theorems~\ref{thm:2sidedidealproduct} and~\ref{thm:2sidedidealboolean}.
\end{proof}

\section{Conclusions}

Two complete DFAs over different alphabets $\Sig'$ and $\Sig$ are incomplete DFAs over  $\Sig' \cup \Sig$. Each DFA can be completed by adding  an empty state and sending all transitions induced by letters not in the DFA's alphabet to that state.
This results in an $(m+1)$-state DFA and an $(n+1)$-state DFA. 
We have shown that the tight bounds for boolean operations are
$(m+1)(n+1)$ for union and symmetric difference, $mn+m$ for difference, and $mn$ for intersection, while
the tight bound for product is $m2^n+2^{n-1}$. In the same-alphabet case the tight bound is $mn$ for all boolean operations and it is $(m-1)2^n+2^{n-1}$ for product. 
 In the case of all three types of ideals, the unrestricted bounds for boolean operations are the same as those for arbitrary regular languages.
The bounds for product are higher that in the restricted case for all three kinds of ideals.

In summary,  the restriction of identical alphabets is unnecessary and leads to incorrect results. 
It  should be noted that if the two languages in question already have empty quotients, then making the alphabets the same does not require the addition of any states, and the traditional same-alphabet methods are correct. This is the case, for example, for prefix-free, suffix-free and finite languages.
\medskip



\providecommand{\noopsort}[1]{}

\end{document}